\documentclass[12pt]{article}
\usepackage{a4wide}
\usepackage{latexsym}
\usepackage{amsmath}
\usepackage{amsfonts}
\usepackage{amssymb}
\usepackage{amscd}
\usepackage{amsthm}
\usepackage{cite}
\usepackage{placeins}

\usepackage{pslatex}
\usepackage{graphicx}
\usepackage[latin1,utf8]{inputenc}
\usepackage[T1]{fontenc}

\allowdisplaybreaks

\newcommand{\bq}{\begin{eqnarray}}
\newcommand{\eq}{\end{eqnarray}}
\newcommand{\eps}{\varepsilon}

\newcommand{\Eulerconstant}{\gamma_{\mathrm{E}}}

\newcommand{\NV}{n}
\newcommand{\extraassumption}{5}

\theoremstyle{plain}
\newtheorem{theoremcounter}{}[]

\newtheorem{theorem}[theoremcounter]{Theorem}

\newtheorem{proposition}[theoremcounter]{Proposition}
\newtheorem{lemma}[theoremcounter]{Lemma}

\numberwithin{equation}{section}

\begin{document}

\thispagestyle{empty}

\begin{flushright}
  MITP/20-007
% \\ version of \today
\end{flushright}

\vspace{1.5cm}

\begin{center}
  {\Large\bf On the computation of intersection numbers for twisted cocycles \\
  }
  \vspace{1cm}
  {\large Stefan Weinzierl\footnote{weinzierl@uni-mainz.de} \\
  \vspace{1cm}
      {\small \em PRISMA Cluster of Excellence, Institut f{\"u}r Physik, }\\
      {\small \em Johannes Gutenberg-Universit{\"a}t Mainz,}\\
      {\small \em D - 55099 Mainz, Germany}\\
  } 
\end{center}

\vspace{2cm}

% abstract ---------------------------------------
\begin{abstract}\noindent
  {
Intersection numbers of twisted cocycles arise in mathematics in the field of algebraic geometry.
Quite recently, they appeared in physics:
Intersection numbers of twisted cocycles define a scalar product on the vector space of Feynman integrals.
With this application, the practical and efficient computation of intersection numbers of twisted cocycles 
becomes a topic of interest.
An existing algorithm for the computation of intersection numbers of twisted cocycles requires in intermediate
steps the introduction of algebraic extensions (for example square roots), 
although the final result may be expressed without algebraic extensions.
In this article I present an improvement of this algorithm,
which avoids algebraic extensions.
   }
\end{abstract}

\vspace*{\fill}

% -----------------------------------------------------------------------------
\newpage

\section{Introduction}
\label{sect:intro}

Intersection numbers of twisted cocycles arise in mathematics in the field of algebraic geometry
and have been investigated there \cite{aomoto1975,Matsumoto:1994,cho1995,matsumoto1998,Ohara:2003,Goto:2013,Goto:2015aaa,Goto:2015aab,Goto:2015aac,Matsubara-Heo:2019,Aomoto:book,Yoshida:book}.
Quite recently, it has been become clear that they are also relevant to physics
and they provide an underlying mathematical framework for some established formulae and methods.
First of all the Cachazo-He-Yuan formula \cite{Cachazo:2013gna,Cachazo:2013hca,Cachazo:2013iea}
for tree-level scattering amplitude may be interpreted 
as an intersection number \cite{Mizera:2017rqa,Mizera:2017cqs,Mizera:2019gea,Mizera:2019blq} 
in the case where both half-integrands have only simple poles.
Secondly, there is an interesting application in the context of Feynman integrals:
Intersection numbers can be used to define an inner product on the space of 
master integrals \cite{Mastrolia:2018uzb,Frellesvig:2019kgj,Frellesvig:2019uqt,Mizera:2019vvs,Chen:2020uyk,Frellesvig:2020qot,Caron-Huot:2021xqj}.
This gives an alternative to Feynman integral reduction, 
traditionally done with the help of integration-by-parts identities \cite{Tkachov:1981wb,Chetyrkin:1981qh}.
This raises the question if the use of intersection numbers can help to speed-up the task of Feynman integral
reduction.
In a first step this requires an algorithm for the efficient calculation of intersection numbers of twisted cocycles.
This is the topic of this paper.

An existing algorithm \cite{mimachi2004,Mizera:2019gea,Frellesvig:2019uqt}
for the computation of multivariate intersection numbers of twisted cocycles
uses a recursive approach. At each step, a sum over the residues of all singular points of a matrix is performed.
The singular points are given by polynomial equations and this step introduces in general algebraic extensions
(e.g. roots).

On the other hand it is well-known that integration-by-parts reduction can be done entirely with polynomials
and does not introduce algebraic extensions.

It is therefore of interest to investigate if multivariate intersection numbers can be computed without
introducing algebraic extensions in intermediate stages.
Analysing the Cachazo-He-Yuan formula shows a possible path: The original Cachazo-He-Yuan formula
involves a sum over residues and evaluating the residues individually inevitably leads to 
algebraic extensions \cite{Weinzierl:2014vwa}.
However, the sum of all residues is a global residue and can be evaluated 
without algebraic extensions \cite{Sogaard:2015dba,Bosma:2016ttj,Cattani:2005,Zhang:2016kfo}.

The Cachazo-He-Yuan formula specialised to the bi-adjoint scalar theory with half-integrands given by Parke-Taylor factors
has only simple poles.
It is a rather simple intersection number, where all polynomials are hyperplanes.
In the application towards Feynman integrals this will no longer be true and we will encounter more general hypersurfaces.
Let us mention that in the case where all polynomials are hyperplanes and the cocycles have only simple poles,
ref.~\cite{Mizera:2017rqa} relates the multivariate intersection number to a sum of residues over the critical points
of the connection. This sum is a global residue and can be evaluated without algebraic extensions.

It is worth pointing out the difference between the Cachazo-He-Yuan formula and the inner product for Feynman integrals
with respect to intersection numbers and global residues:
The Cachazo-He-Yuan formula is always a global residue. If both half-integrand have simple poles, it is also an intersection number.
The inner product for Feynman integrals is always an intersection number. If at all stages we only have simple poles, it can be computed 
from global residues.

Thus the task is to find an algorithm for the computation of intersection numbers, which
avoids algebraic extensions and is not restricted to hyperplanes and simple poles.
In this paper I present such an algorithm.
The algorithm consists of three steps:
\begin{enumerate}
\item Recursive approach: The algorithm integrates out one variable at a time.
This part is identical to the algorithm of \cite{Mizera:2019gea,Frellesvig:2019uqt}.
It has the advantage to reduce a multivariate problem to a univariate problem.
\item Reduction to simple poles: In general we deal in cohomology with equivalence classes.
We may replace a representative of an equivalence class with higher poles with an equivalent representative
with only simple poles.
This is similar to integration-by-part reduction. However, let us stress that the involved systems of linear
equations are usually significantly smaller compared to standard integration-by-part reduction.
\item Evaluation of the intersection number as a global residue. Having reduced our objects to simple poles,
we may evaluate the intersection in one variable as an univariate global residue.
This is easily computed and does not involve algebraic extensions.
\end{enumerate}
This paper is organised as follows:
In section~\ref{sect:notation} we introduce our notation and describe the basic set-up.
In section~\ref{sect:recursive} we review the recursive approach for the computation 
of a multivariate intersection number.
In section~\ref{sect:equivalence} we discuss the equivalence classes of the coefficients,
when an $n$-dimensional cocycle is expanded in a basis of the $(n-1)$-dimensional cohomology group.
The coefficients may have higher poles in the $n$-th variable.
In section~\ref{sect:reduction} we show how the pole order can be reduced systematically.
Section~\ref{sect:forumula} contains the main result of this paper:
It gives a formula for the intersection number of the coefficients in the case of simple poles.
Section~\ref{sect:global_residue} is dedicated to the efficient computation of an univariate global
residue.
Although it is not the main topic of this paper, we discuss in section~\ref{sect:bases}
briefly how bases of twisted cohomology groups / bases of master integrals are obtained.
In section~\ref{sect:algorithm} we summarise the algorithm for the computation of intersection numbers.
A few examples are given in section~\ref{sect:examples}.
Section~\ref{sect:applications} discusses the application towards Feynman integrals.
Finally, our conclusions are given in section~\ref{sect:conclusions}.
In appendix~\ref{sect:laurent} we summarise the algorithm of \cite{Mizera:2019gea,Frellesvig:2019uqt}.

% -----------------------------------------------------------------------------

\section{Notation and definitions}
\label{sect:notation}

Let ${\mathbb K}$ be a field.
In typical applications we have ${\mathbb K}={\mathbb Q}$
or ${\mathbb K}={\mathbb Q}(y_1,\dots,y_s)$.
Consider $m$ polynomials $p_i$ in $n$ variables $z=(z_1,\dots,z_n)$:
\bq
\label{def_polynomials}
 p_i & \in & {\mathbb K}\left[z_1,\dots,z_n\right],
 \;\;\;\;\;\;\;\;\;
 1 \le i \le m.
\eq
For $m$ complex numbers $\gamma=(\gamma_1,\dots,\gamma_m)$ 
we set
\bq
\label{def_u}
 u
 & = &
 \prod\limits_{i=1}^m p_i^{\gamma_i},
\eq
and
\bq
 \omega
 & = &
 d \ln u
 \; = \;
 \sum\limits_{j=1}^n  \omega_j dz_j,
 \nonumber \\
 \omega_j & = & \frac{\partial \ln u}{\partial z_j} \; = \; \frac{P_j}{Q_j},
 \;\;\;\;\;\;
 P_j, Q_j \; \in \; {\mathbb K}\left[z_1,\dots,z_n\right],
 \;\;\;\;\;\;
 \gcd\left(P_j,Q_j\right) \; = \; 1.
\eq
The differential one-form $\omega$ defines a connection and a covariant derivative
\bq
 \nabla_\omega & =& d + \omega.
\eq
$\omega$ is also called the ``twist''.
Set 
\bq
 D_i \; = \; \{ p_i = 0 \} \; \subset \; {\mathbb C}^n
 & \mbox{and} & 
 D \; = \; 
 \bigcup\limits_{i=1}^m D_i.
\eq
Points $z^{\mathrm{crit}}=(z^{\mathrm{crit}}_1,\dots,z^{\mathrm{crit}}_n)$ which satisfy
\bq
 P_1 \; = \; \dots \; = \; P_n \; = \; 0
\eq
are called critical points.
A critical point $z^{\mathrm{crit}}$ is called proper, if
\bq
 z^{\mathrm{crit}} & \notin & D.
\eq
A critical point $z^{\mathrm{crit}}$ is non-degenerate if the Hessian matrix
\bq
 H_{ij}\left(z\right)
 & = &
 \frac{\partial^2 u}{\partial z_i \partial z_j}
\eq
is invertible at $z=z^{\mathrm{crit}}$.
We consider rational differential $n$-forms $\varphi$ in the variables $z=(z_1,\dots,z_n)$, 
which are holomorphic on ${\mathbb C}^n - D$.
The rational $n$-forms $\varphi$ are of the form
\bq
\label{representative_left}
 \varphi
 & = &
 \frac{q}{p_1^{n_1} \dots p_m^{n_m}} \; dz_n \wedge \dots \wedge dz_1,
 \;\;\;\;\;\;\;\;\;
 q \in {\mathbb K}\left[z_1,\dots,z_n\right],
 \;\;\;
 n_i \in {\mathbb N}_0.
\eq
Using the reversed wedge product $dz_n \wedge \dots \wedge dz_1$ is at this stage just a convention.
Two $n$-forms $\varphi'$ and $\varphi$ are called equivalent,
if they differ by a covariant derivative
\bq
 \varphi'
 \sim
 \varphi
 & \Leftrightarrow &
 \varphi'
 \; = \;
  \varphi
  + \nabla_\omega \xi
\eq
for some $(n-1)$-form $\xi$.
We denote the equivalence classes by $\langle \varphi |$.
Being $n$-forms, each $\varphi$ is closed with respect to $\nabla_\omega$ and the equivalence classes
define the twisted cohomology group $H^n_\omega$:
\bq
 \left\langle \varphi \right|
 & \in &
 H^n_\omega.
\eq
Under certain assumptions it can be shown \cite{aomoto1975} that 
the twisted cohomology groups $H^k_\omega$ vanish for $k \neq n$, thus $H^n_\omega$ is the only interesting
twisted cohomology group.
The physical interpretation for Feynman integrals is discussed in \cite{Caron-Huot:2021xqj}.

The dual twisted cohomology group is given by
\bq
 \left( H^n_\omega \right)^\ast 
 & = &
 H^n_{-\omega}.
\eq
Elements of $( H^n_\omega )^\ast$ are denoted by $| \varphi \rangle$.
We have
\bq
 \left| \varphi' \right\rangle
 =
 \left| \varphi \right\rangle
 & \Leftrightarrow &
 \varphi'
 \; = \;
  \varphi
  + \nabla_{-\omega} \xi
\eq
for some $(n-1)$-form $\xi$.
A representative of a dual cohomology class is of the form
\bq
\label{representative_right}
 \varphi
 & = &
 \frac{q}{p_1^{n_1} \dots p_m^{n_m}} \; dz_1 \wedge \dots \wedge dz_n,
 \;\;\;\;\;\;\;\;\;
 q \in {\mathbb K}\left[z_1,\dots,z_n\right],
 \;\;\;
 n_i \in {\mathbb N}_0.
\eq
It will be convenient to use here the order $dz_1 \wedge \dots \wedge dz_n$ in the wedge product.

For a $n$-form $\varphi_L$ and a $n$-form $\varphi_R$ we define the rational functions 
$\hat{\varphi}_L$ and $\hat{\varphi}_R$ by stripping off $dz_n \wedge \dots \wedge dz_1$ or $dz_1 \wedge \dots \wedge dz_n$, respectively.
\bq
 \varphi_L 
 \; = \; 
 \hat{\varphi}_L dz_n \wedge \dots \wedge dz_1,
 & &
 \varphi_R
 \; = \; 
 \hat{\varphi}_R dz_1 \wedge \dots \wedge dz_n.
\eq
The central object of this article are the intersection numbers
\bq
 \left\langle \varphi_L \right. \left| \varphi_R \right\rangle,
 \;\;\;\;\;\;\;\;\;
 \left\langle \varphi_L \right| \in H^n_\omega,
 \;\;\;\;\;\;
 \left| \varphi_R \right\rangle \in \left( H^n_\omega \right)^\ast.
\eq
They are defined by \cite{cho1995,Aomoto:book}
\bq
 \left\langle \varphi_L \right. \left| \varphi_R \right\rangle
 & = &
 \frac{1}{\left(2\pi i\right)^n}
 \int \iota_\omega\left(\varphi_L\right) \wedge \varphi_R
 \; = \;
 \frac{1}{\left(2\pi i\right)^n}
 \int \varphi_L \wedge \iota_{-\omega}\left(\varphi_R\right),
\eq
where $\iota_\omega$ maps $\varphi_L$ to its compactly supported version,
and similar for $\iota_{-\omega}$.
From the definition we have
\bq
 \left\langle \varphi_L \right. \left| \varphi_R \right\rangle_{\omega}
 & = &
 \left(-1\right)^n
 \left\langle \varphi_R \right. \left| \varphi_L \right\rangle_{-\omega}.
\eq
We are interested in evaluating this integral.
In ref.~\cite{Mizera:2019gea,Frellesvig:2019uqt} a recursive algorithm for the evaluation of multivariate intersection
numbers has been given. This algorithm is briefly reviewed in appendix~\ref{sect:laurent}.
This algorithm requires in intermediate steps algebraic extensions (the roots of the polynomials $p_i$ in the variable $z_j$), although in the final expressions the roots drop out.
It is therefore desirable to have an algorithm which computes the intersection numbers without the need
of introducing algebraic extensions.
In this article I present such an algorithm.

As in \cite{Mizera:2019gea,Frellesvig:2019uqt}, we have to make some assumptions.
Standard assumptions related to the connection one-form $\omega$ are:
\begin{enumerate}
\item We require that the exponents $\gamma_1,\dots,\gamma_m$ are generic, in particular not integers.
\item We require that there are only a finite number of proper critical points, all of which are non-degenerate.
\end{enumerate}
The algorithm of \cite{Mizera:2019gea,Frellesvig:2019uqt} assumes that there is a suitable non-singular sequence of fibrations, from
which the intersection number can be computed recursively.
This is also an assumption of our algorithm.
In technical terms, this implies (the definitions of the quantities will be given in the next section)
\begin{enumerate}
\setcounter{enumi}{2}
\item At each step in the recursion and in every punctured neighbourhood of a singular point $z_i=z^{\mathrm{sing}}_i$ 
there are unique holomorphic vector-valued solutions  
$\hat{\psi}^{({\bf i})}_{L,j}$ and $\hat{\psi}^{({\bf i})}_{R,j}$ of
\bq
\label{holomorphic_solution}
 \partial_{z_i} \hat{\psi}^{({\bf i})}_{L,j} + \hat{\psi}^{({\bf i})}_{L, k} \Omega^{({\bf i})}_{k j} \; = \; \hat{\varphi}^{({\bf i})}_{L,j},
 & &
 \partial_{z_i} \hat{\psi}^{({\bf i})}_{R,j} - \Omega^{({\bf i})}_{j k} \hat{\psi}^{({\bf i})}_{R, k} \; = \; \hat{\varphi}^{({\bf i})}_{R,j}.
\eq
\end{enumerate}
In addition we will assume that 
\begin{enumerate}
\setcounter{enumi}{3}
\item there are bases of $H^{({\bf 0})}_\omega, \dots, H^{({\bf n-1})}_\omega$ such that
the connection matrices $\Omega^{({\bf 1})}, \dots, \Omega^{({\bf n})}$ have only simple poles,
\item the determinant
\bq
 \det\left(\Omega^{({\bf i})}\right)
\eq
has $\nu_{\bf i}=\dim H^{({\bf i})}_\omega$ critical points in the variable $z_i$.
\end{enumerate}
Assumption (4) is required for the reduction to simple poles.
We will comment on assumption (\extraassumption) in section~\ref{sect:example_elliptic_curve_1} and section~\ref{sect:equal_mass_sunrise}.

% -----------------------------------------------------------------------------

\section{The recursive structure}
\label{sect:recursive}

We will compute the intersection numbers in $n$ variables $z_1, \dots z_n$ recursively by splitting
the problem into the computation of an intersection number in $(n-1)$ variables $z_1, \dots, z_{n-1}$
and the computation of a (generalised) intersection number in the variable $z_n$.
By recursion, we therefore have to compute only (generalised) intersection numbers in a single variable $z_i$.
This reduces the multivariate problem to an univariate problem.
This step is essentially identical to \cite{mimachi2004,Mizera:2019gea,Frellesvig:2019uqt}.

Let us comment on the word ``generalised'' intersection number:
We only need to discuss the univariate case.
Consider two cohomology classes $\langle \varphi_L |$ and $| \varphi_R \rangle$.
Representatives $\varphi_L$ and $\varphi_R$ for the two cohomology classes
$\langle \varphi_L |$ and $| \varphi_R \rangle$ are differential one-forms and of the form as
in eq.~(\ref{representative_left}) or eq.~(\ref{representative_right}).
We may view the representatives $\varphi_L$ and $\varphi_R$,
the cohomology classes $\langle \varphi_L |$ and $| \varphi_R \rangle$, and the twist $\omega$ 
as scalar quantities.

Consider now a vector of $\nu$ differential one-forms $\varphi_{L,j}$ in the variable $z$, 
where $j$ runs from $1$ to $\nu$.
Similar, consider for the dual space a $\nu$-dimensional vector $\varphi_{R,j}$
and generalise $\omega$ to a $(\nu\times \nu)$-dimensional matrix $\Omega$.
The equivalence classes $\langle \varphi_{L,j} |$ and $| \varphi_{R,j} \rangle$ are now defined
by
\bq
 \hat{\varphi}_{L,j}' \; = \; \hat{\varphi}_{L,j} + \partial_z \xi_j + \xi_i \Omega_{i j}
 & \mbox{and} &
 \hat{\varphi}_{R,j}' \; = \; \hat{\varphi}_{R,j} + \partial_z \xi_j - \Omega_{j i} \xi_i,
\eq
for some zero-forms $\xi_j$ (i.e. functions).
We will define intersection numbers for the vector-valued cohomology classes
$\langle \varphi_{L,j} |$ and $| \varphi_{R,j} \rangle$.

Readers familiar with gauge theories will certainly recognise that the generalisation is exactly the same
step as going from an Abelian gauge theory (like QED) to a non-Abelian gauge theory (like QCD).

Let us now set up the notation for the recursive structure.
We fix an ordered sequence $(z_{\sigma_1},\dots,z_{\sigma_n})$, indicating that we first integrate out $z_{\sigma_1}$,
then $z_{\sigma_2}$, etc..
Without loss of generality we will always consider the order $(z_1,\dots,z_n)$, unless indicated otherwise.

For $i=0,\dots,n$ we consider a fibration $E_i : {\mathbb C}^n \rightarrow B_i$
with total space ${\mathbb C}^n$,
fibre $V_i = {\mathbb C}^i$ parametrised by the coordinates $(z_1,\dots,z_i)$
and base $B_i = {\mathbb C}^{n-i}$ parametrised by the coordinates $(z_{i+1},\dots,z_n)$.
The covariant derivative splits as
\bq
 \nabla_\omega 
 & = &
 \nabla_\omega^{({\bf i}),F} + \nabla_\omega^{({\bf i}),B},
\eq
with
\bq
 \nabla_\omega^{({\bf i}),F}
 \; = \;
 \sum\limits_{j=1}^i dz_j \left( \frac{\partial}{\partial z_j} + \omega_j \right),
 & &
 \nabla_\omega^{({\bf i}),B}
 \; = \;
 \sum\limits_{j=i+1}^n dz_j \left( \frac{\partial}{\partial z_j} + \omega_j \right).
\eq
One sets
\bq
 \omega^{({\bf i})}
 & = &
 \sum\limits_{j=1}^i \omega_j dz_j.
\eq
Clearly, for $i=n$ we have
\bq
 \omega^{({\bf n})}
 \; = \; 
 \omega, 
 & &
 \nabla_\omega^{({\bf n}),F}
 \; = \;
 \nabla_\omega.
\eq
Following \cite{Frellesvig:2019uqt}, we study for each $i$ the twisted cohomology group in the fibre,
defined by replacing $\omega$ with $\omega^{({\bf i})}$.
The additional variables $(z_{i+1},\dots,z_n)$ are treated as parameters in the same way as
the variables $(y_1,\dots,y_s)$ of the ground field ${\mathbb K}={\mathbb Q}(y_1,\dots,y_s)$.
For each $i$ only the $i$-th cohomology group is of interest and 
for simplicity we write
\bq
 H^{({\bf i})}_\omega
 \; = \;
 H^{i}_{\omega^{({\bf i})}},
 & &
 \left( H^{({\bf i})}_\omega \right)^\ast
 \; = \;
 \left( H^{i}_{\omega^{({\bf i})}} \right)^\ast.
\eq
We denote the dimensions of the twisted cohomology groups by
\bq 
 \nu_{\bf i} & = & \dim H^{({\bf i})}_\omega \; = \; \dim \left( H^{({\bf i})}_\omega \right)^\ast.
\eq
Let 
$\langle e^{({\bf i})}_j |$ with $1\le j \le \nu_{\bf i}$
be a basis of $H^{({\bf i})}_\omega$ and let
$| h^{({\bf i})}_j \rangle$ with $ 1 \le j \le \nu_{\bf i}$
be a basis of $( H^{({\bf i})}_\omega )^\ast$.
We denote the $(\nu_{\bf i} \times \nu_{\bf i})$-dimensional intersection matrix by $C_{\bf i}$.
The entries are given by
\bq
 \left( C_{\bf i} \right)_{j k}
 & = &
 \left\langle e^{({\bf i})}_{j} \right| \left. h^{({\bf i})}_{k} \right\rangle.
\eq
The matrix $C_{\bf i}$ is invertible.
Given a basis $\langle e^{({\bf i})}_j |$ of $H^{({\bf i})}_\omega$
we say that a basis $| d^{({\bf i})}_j \rangle$ of $( H^{({\bf i})}_\omega )^\ast$ is
the dual basis with respect to $\langle e^{({\bf i})}_j |$
if
\bq
 \left\langle e^{({\bf i})}_{j} \right| \left. d^{({\bf i})}_{k} \right\rangle
 & = &
 \delta_{j k}.
\eq
We may always construct a dual basis:
\bq
\label{construction_dual_basis}
 \left| d^{({\bf i})}_j \right\rangle
 & = &
 \left| h^{({\bf i})}_k \right\rangle
 \left( C_{\bf i}^{-1} \right)_{kj}.
\eq
Using the dual basis will simplify some of the formulae in the sequel.

The essential step in the recursive approach is to expand 
the twisted cohomology class $\langle \varphi^{({\bf n})}_L | \in H^{({\bf n})}_\omega$ 
in the basis of $H^{({\bf n-1})}_\omega$:
\bq
 \left\langle \varphi^{({\bf n})}_L \right|
 & = &
 \sum\limits_{j=1}^{\nu_{\bf n-1}}
 \left\langle \varphi^{({\bf n})}_{L,j} \right| \wedge \left\langle e^{({\bf n-1})}_j \right|.
\eq
Here, $\langle e^{({\bf n-1})}_j |$ denotes a basis of $H^{({\bf n-1})}_\omega$.
Representatives of these cohomology classes are differential $(n-1)$-forms of the form
\bq
\label{def_representative_e_n_minus_1}
 \hat{e}^{({\bf n-1})}_j dz_{n-1} \wedge \dots dz_1,
\eq
where $\hat{e}^{({\bf n-1})}_j$ may depend on all variables $(z_1, \dots, z_n)$.
On the other hand, the coefficients $\langle \varphi^{({\bf n})}_{L,j}|$ are one-forms proportional to $dz_n$.
They only depend on $z_n$, but not on $(z_1,\dots,z_{n-1})$.
The coefficients $\langle \varphi^{({\bf n})}_{L,j}|$
are given by
\bq
\label{def_coeff_left}
 \left\langle \varphi^{({\bf n})}_{L,j} \right|
 & = &
 \left\langle \varphi^{({\bf n})}_L \left| d^{({\bf n-1})}_{j} \right. \right\rangle.
\eq
Note that the coefficients $\langle \varphi^{({\bf n})}_{L,j}|$ are obtained by computing only
intersection numbers in $(n-1)$ variables.
This is compatible with the recursive approach.
It also shows that the coefficients do not depend on the variables $(z_1,\dots,z_{n-1})$, as these
variables are integrated out.
Given a representative $\varphi^{({\bf n})}_L$ of the class $\langle \varphi^{({\bf n})}_L |$
and representatives $d^{({\bf n-1})}_j$ of the basis elements $| d^{({\bf n-1})}_j \rangle$ of $(H^{({\bf n-1})}_\omega)^\ast$ 
we may (unambiguously) compute a representative $\hat{\varphi}^{({\bf n})}_{L,j} dz_n$ for the coefficients $\langle \varphi^{({\bf n})}_{L,j} |$ through
eq.~(\ref{def_coeff_left}).
The result will not depend on which representatives $d^{({\bf n-1})}_j$ we choose for the basis
$| d^{({\bf n-1})}_j \rangle$ of $(H^{({\bf n-1})}_\omega)^\ast$, 
the $(n-1)$-fold intersection number in eq.~(\ref{def_coeff_left})
is invariant under redefining individual $d^{({\bf n-1})}_j$ by
$\nabla_{-\omega^{({\bf n-1})}} \psi$ for some $(n-2)$-form $\psi$ 
such that $\nabla_{-\omega^{({\bf n-1})}} \psi$ is proportional to $dz_1 \wedge \dots \wedge dz_{n-1}$.
Eq.~(\ref{def_coeff_left}) is also invariant under redefining $\varphi^{({\bf n})}_L$
by
\bq
 f\left(z_n\right) dz_n \wedge \nabla_{\omega^{({\bf n-1})}} \psi
\eq
for an arbitrary function $f(z_n)$ and a $(n-2)$-form $\psi$ as above.
However, $\varphi^{({\bf n})}_L$ represents a larger equivalence class, invariant under
\bq
 \varphi^{({\bf n})}_L
 & \rightarrow &
 \varphi^{({\bf n})}_L
 + \nabla_\omega \xi
\eq
for some $(n-1)$-form $\xi$.
This has the effect that the representatives of the coefficients are not unique
and we should also think of the coefficients as equivalence classes (hence the notation $\langle \varphi^{({\bf n})}_{L,j}|$).
In the next section we will discuss in detail the freedom in redefining the coefficients.
In general we cannot redefine a single coefficient $\langle \varphi^{({\bf n})}_{L,j_{\mathrm{fix}}}|$ for a fixed $j_{\mathrm{fix}}$, 
but we have to consider the $\nu_{\bf n-1}$-dimensional
vector of all coefficients $\langle \varphi^{({\bf n})}_{L,j}|$ (with $j=1,.\dots,\nu_{\bf n-1}$).

We close this paragraph by giving the corresponding formulae for the dual twisted cohomology classes.
One expands $| \varphi^{({\bf n})}_R \rangle \in ( H^{({\bf n})}_\omega )^\ast$ in the dual basis of $( H^{({\bf n-1})}_\omega )^\ast$:
\bq
 \left| \varphi^{({\bf n})}_R \right\rangle
 & = &
 \sum\limits_{j=1}^{\nu_{\bf n-1}}
 \left| d^{({\bf n-1})}_j \right\rangle
 \wedge
 \left| \varphi^{({\bf n})}_{R,j} \right\rangle.
\eq
The coefficients $| \varphi^{({\bf n})}_{R,j}\rangle$ are one-forms proportional to $dz_n$ and 
independent of $(z_1,\dots,z_{n-1})$
They are given by
\bq
 \left| \varphi^{({\bf n})}_{R,j} \right\rangle
 & = &
 \left\langle \left. e^{({\bf n-1})}_{j} \right| \varphi^{({\bf n})}_R \right\rangle.
\eq
Please note that we have chosen the dual basis which satisfies
\bq
 \left\langle \left. e^{({\bf n-1})}_{j} \right| d^{({\bf n-1})}_k \right\rangle
 & = &
 \delta_{j k}.
\eq

% -----------------------------------------------------------------------------

\section{The equivalence class of the coefficients}
\label{sect:equivalence}

In this section we study in detail the equivalence classes of the coefficients
$\langle \varphi^{({\bf n})}_{L,j}|$ and $| \varphi^{({\bf n})}_{R,j}\rangle$.
We will see that they transform as vectors.

\begin{proposition}
\label{proposition_1}
Consider the cohomology class $\langle \varphi^{({\bf n})}_L | \in H^{({\bf n})}_\omega$
and expand $\langle \varphi^{({\bf n})}_L |$ 
in the basis of $H^{({\bf n-1})}_\omega$:
\bq
 \left\langle \varphi^{({\bf n})}_L \right|
 & = &
 \sum\limits_{j=1}^{\nu_{\bf n-1}}
 \left\langle \varphi^{({\bf n})}_{L,j} \right| \wedge \left\langle e^{({\bf n-1})}_j \right|.
\eq
Define $\hat{\varphi}^{({\bf n})}_{L,j}$ by $\varphi^{({\bf n})}_{L,j} = \hat{\varphi}^{({\bf n})}_{L,j} dz_n$.
Changing the representative amounts to
\bq
\label{gauge_trafo_left}
 \left\langle \varphi^{({\bf n})}_L \right|
 & \rightarrow &
 \left\langle \varphi^{({\bf n})}_L \right|
 + 
 \left\langle \nabla_\omega \xi \right|,
\eq
for some $(n-1)$-form $\xi$.
Let us now consider transformations which are generated by $(n-1)$-forms $\xi$ 
of the type
\bq
 \xi 
 & = &
 \sum\limits_{j=1}^{\nu_{\bf n-1}}
 f_j\left(z_n\right) \left\langle e^{({\bf n-1})}_j \right|,
\eq
where the functions $f_j(z_n)$ depend only on $z_n$, but not on $z_1,\dots,z_{n-1}$.
Then the representatives of the coefficients transform as
\bq
\label{left_invariance}
 \hat{\varphi}^{({\bf n})}_{L,j}
 & \rightarrow &
 \hat{\varphi}^{({\bf n})}_{L,j}
 +
 f_i \left( \overleftarrow{\partial}_{z_n} \delta_{i j} + \Omega^{({\bf n})}_{ij} \right),
\eq
where the $(\nu_{\bf n-1} \times \nu_{\bf n-1})$-matrix $\Omega^{({\bf n})}$ is defined by
\bq
\label{def_Omega_left}
 \Omega^{({\bf n})}_{ij}
 & = & 
 \left\langle \left(\partial_{z_n}+\omega_{n}\right) e^{({\bf n-1})}_i \right| \left. d^{({\bf n-1})}_{j} \right\rangle.
\eq
\end{proposition}
\begin{proof}
The definition of $\Omega^{({\bf n})}$ in eq.~(\ref{def_Omega_left}) implies that
\bq
 \left\langle \left(\partial_{z_n}+\omega_{n}\right) e^{({\bf n-1})}_i \right|
 & = &
 \Omega^{({\bf n})}_{ij}
 \left\langle e^{({\bf n-1})}_{j} \right|,
\eq
and with eq.~(\ref{def_representative_e_n_minus_1}) we have
\bq
 \left\langle \nabla_\omega \xi \right|
 & = &
 \sum\limits_{j=1}^{\nu_{\bf n-1}}
 \left( \partial_{z_n}f_j + f_i \Omega^{({\bf n})}_{ij} \right) dz_n \wedge \left\langle e^{({\bf n-1})}_j \right|.
\eq
The claim follows.
\end{proof}
In physics, we think of eq.~(\ref{gauge_trafo_left}) as a gauge transformation.
The transformation properties of the representatives of the coefficients $| \varphi^{({\bf n})}_{R,j}\rangle$
are similar:
\begin{proposition}
\label{proposition_2}
Consider the cohomology class $| \varphi^{({\bf n})}_R \rangle \in (H^{({\bf n})}_\omega)^\ast$
and expand $| \varphi^{({\bf n})}_R \rangle$ 
in the dual basis $| d^{({\bf n-1})}_j \rangle$
of $( H^{({\bf n-1})}_\omega )^\ast$:
\bq
 \left| \varphi^{({\bf n})}_R \right\rangle
 & = &
 \sum\limits_{j=1}^{\nu_{\bf n-1}}
 \left| d^{({\bf n-1})}_j \right\rangle
 \wedge
 \left| \varphi^{({\bf n})}_{R,j} \right\rangle,
\eq
Define $\hat{\varphi}^{({\bf n})}_{R,j}$ by $\varphi^{({\bf n})}_{R,j} = \hat{\varphi}^{({\bf n})}_{R,j} dz_n$.
Changing the representative amounts to
\bq
\label{gauge_trafo_right}
 \left| \varphi^{({\bf n})}_R \right\rangle
 & \rightarrow &
 \left| \varphi^{({\bf n})}_R \right\rangle
 + 
 \left| \nabla_{-\omega} \xi \right\rangle,
\eq
for some $(n-1)$-form $\xi$.
Let us now consider transformations which are generated by $(n-1)$-forms $\xi$ 
of the type
\bq
 \xi
 & = &
 \left(-1\right)^{n-1}
 \sum\limits_{j=1}^{\nu_{\bf n-1}}
 f_j\left(z_n\right)
 \left| d^{({\bf n-1})}_j \right\rangle,
\eq
where the functions $f_j(z_n)$ depend only on $z_n$, but not on $z_1,\dots,z_{n-1}$.
Then the representatives of the coefficients transform as
\bq
\label{right_invariance}
 \hat{\varphi}^{({\bf n})}_{R,j}
 & \rightarrow &
 \hat{\varphi}^{({\bf n})}_{R,j}
 +
 \left( \delta_{j k} \partial_{z_n} - \Omega^{({\bf n})}_{jk} \right) f_k,
\eq
where the $(\nu_{\bf n-1} \times \nu_{\bf n-1})$-matrix $\Omega^{({\bf n})}$ is now defined by
\bq
\label{def_Omega_right}
 \Omega^{({\bf n})}_{ij}
 & = & 
 -
 \left\langle e^{({\bf n-1})}_{i} \right| \left. \left(\partial_{z_n}-\omega_{n}\right) d^{({\bf n-1})}_j  \right\rangle.
\eq
\end{proposition}
\begin{proof}
The definition of $\Omega^{({\bf n})}$ in eq.~(\ref{def_Omega_right}) implies that
\bq
 \left| \left(\partial_{z_n}-\omega_{z_n}\right) d^{({\bf n-1})}_{j} \right\rangle
 & = &
 - \left| d^{({\bf n-1})}_{i} \right\rangle \Omega^{(k)}_{ij}.
\eq
The rest of the proof proceeds as in the case of proposition~\ref{proposition_1}.
\end{proof}
\begin{proposition}
\label{proposition_3}
The definitions of $\Omega^{({\bf n})}$ in eq.~(\ref{def_Omega_left})
and in eq.~(\ref{def_Omega_right}) agree.
\end{proposition}
\begin{proof}
Suppose eq.~(\ref{def_Omega_left}) defines 
a matrix $\Omega^{({\bf n})}_L$ and eq.~(\ref{def_Omega_right}) defines a matrix $\Omega^{({\bf n})}_R$.
Then
\bq
 0 & = &
 \partial_{z_n} \left( \left\langle \left. e^{({\bf n-1})}_{j} \right| d^{({\bf n-1})}_k \right\rangle \right)
 \; = \; 
 \left\langle \left. \left( \partial_{z_n} + \omega_n \right) e^{({\bf n-1})}_{j} \right| d^{({\bf n-1})}_k \right\rangle
 +
 \left\langle \left. e^{({\bf n-1})}_{j} \right| \left( \partial_{z_n} - \omega_n \right) d^{({\bf n-1})}_k \right\rangle
 \nonumber \\
 & = &
 \Omega^{({\bf n})}_{L,ji}
 \left\langle \left. e^{({\bf n-1})}_{i} \right| d^{({\bf n-1})}_k \right\rangle
 -
 \left\langle \left. e^{({\bf n-1})}_{j} \right| d^{({\bf n-1})}_i \right\rangle
 \Omega^{({\bf n})}_{R,ik}
 \; = \;
 \Omega^{({\bf n})}_{L,jk}
 -
 \Omega^{({\bf n})}_{R,jk}.
\eq
\end{proof}

% -----------------------------------------------------------------------------

\section{Reduction to simple poles}
\label{sect:reduction}

In this section we show how the transformations in eq.~(\ref{left_invariance}) and eq.~(\ref{right_invariance})
can be used to reduce the vector of coefficients
$\langle \varphi^{({\bf n})}_{L,j}|$ and $| \varphi^{({\bf n})}_{R,j}\rangle$
to a form where only simple poles in the variable $z_n$ occur.
In this section we deal at all stages only with univariate rational functions (in the variable $z_n$).
This is a significant simplification compared to the multivariate case.
In particular we may use partial fraction decomposition in the variable $z_n$.
A rational function in the variable $z_n$
\bq
 r\left(z_n\right)
 & = &
 \frac{P\left(z_n\right)}{Q\left(z_n\right)},
 \;\;\;\;\;\;\;\;\;
 P, Q \; \in \; {\mathbb K}\left[z_n\right]
 \;\;\;\;\;\;\;\;\;
 \gcd\left(P,Q\right) \; = \; 1,
\eq
has only simple poles if $\deg P < \deg Q$ and if in the partial fraction decomposition each
irreducible polynomial in the denominator occurs only to power $1$.
The condition $\deg P < \deg Q$ ensures that there are no higher poles at infinity.

In this section we will assume that (i) all entries of $\Omega^{({\bf n})}$ have only simple poles
and (ii) that the linear systems discussed below have a unique solution.

Assumption (i) depends on our choice $\hat{e}^{({\bf n-1})}_j$ for the basis of $H^{({\bf n-1})}_\omega$.
If the initial choice $\hat{e}^{({\bf n-1})}_j$ for the basis of $H^{({\bf n-1})}_\omega$ does not satisfy 
assumption (i) we search for a transformation to a new basis $\hat{e}^{({\bf n-1})}_j{}'$ such
that assumption (i) is satisfied.
This is completely analogous to the way we transform the differential equation for a set of Feynman master integrals.
All techniques can be carried over.
In particular,
using Moser's algorithm \cite{Moser:1959,Lee:2014ioa} 
we may always transform to a new basis $\hat{e}^{({\bf n-1})}_j{}'$
such that $\Omega^{({\bf n})}{}'$ has only simple poles except possibly at one point.
In the context of Feynman integrals we are not aware of an example, where higher poles at a single point remain.

Assumption (ii) boils down to our requirement that the exponents $\gamma_i$ in eq.~(\ref{def_u}) are generic,
in particular non-integer.

It is sufficient to discuss the reduction to simple poles for a $\nu$-dimensional vector $\hat{\varphi}_j$
($1 \le j \le \nu$)
which transforms as
\bq
\label{generic_invariance}
 \hat{\varphi}_{j}
 & \rightarrow &
 \hat{\varphi}_{j}
 +
 \left( \delta_{j k} \partial_{z_n} + \Omega_{jk} \right) f_k.
\eq
The reduction of $\langle \varphi^{({\bf n})}_{L,j}|$ is then achieved by setting $\Omega = (\Omega^{({\bf n})})^T$,
the reduction of $| \varphi^{({\bf n})}_{R,j}\rangle$ is achieved by setting $\Omega = -\Omega^{({\bf n})}$.
In both case we have $\nu=\nu_{\bf n-1}$.

\begin{lemma}
\label{lemma_4}
Assume that $\Omega$ has only simple poles and that
the vector $\hat{\varphi}_{j}$ 
has a pole of order $o > 1$ at infinity.
A gauge transformation with the seed
\bq
 f_j\left(z_n\right) & = & c_j z_n^{o-1},
 \;\;\;\;\;\;\;\;\;
 c_j \; \in \; {\mathbb K}
\eq
reduces the order of the pole at infinity, provided the linear system obtained from 
the condition that the $\nu$ equations
\bq
\label{eq_seed_inf}
 \hat{\varphi}_{j}
 +
 \left( \delta_{j k} \partial_{z_n} + \Omega_{jk} \right) f_k
\eq
have only poles of order $(o-1)$ at infinity, 
yields a solution for the $\nu$ coefficients $c_j$.
Furthermore, this gauge transformation does not introduce higher poles elsewhere.
\end{lemma}
\begin{proof}
Since we assume that $\Omega$ has only simple poles it follows that eq.~(\ref{eq_seed_inf})
has at most a pole of order $o$ at infinity.
We obtain a linear system of equations for the unknown coefficients $c_j$ by partial fraction
decomposition of eq.~(\ref{eq_seed_inf}) for each $j$ and subsequently setting the coefficient
of the monomial term $z_n^{o-2}$ to zero.
By assumption, this linear system of equations has a solution for the coefficients $c_j$.
Having determined the coefficients $c_j$, we define $\hat{\varphi}_{j}'$ by
\bq
 \hat{\varphi}_{j}'
 & = &
 \hat{\varphi}_{j}
 +
 \left( \delta_{j k} \partial_{z_n} + \Omega_{jk} \right) f_k.
\eq
This reduces the order of the pole at infinity by one. 
Note that this procedure may introduce new simple poles at finite points (through $\Omega$).
However, the procedure will never introduce higher poles at finite points (since we assumed that $\Omega$ has 
only simple poles).
\end{proof}

Repeating this procedure we may reduce the pole
at infinity to a simple pole.

Let us give a simple example for the assumption that the linear system must yield a solution.
We consider $\nu=1$, $\Omega = \gamma/z_n$ and $\hat{\varphi}=1$.
$\Omega$ has a simple pole at $z_n=0$ and $z_n=\infty$, $\hat{\varphi}=1$ has a pole of order $2$ at infinity
(and no poles elsewhere).
The seed 
\bq
 f & = & c z_n
\eq
leads to the equation
\bq
 1 + c + c \gamma & = & 0
\eq
and will reduce the pole at infinity, provided $\gamma \neq -1$.
Thus we see that the assumption that the linear system has a solution reduces to the assumption that
the exponent $\gamma$ in eq.~(\ref{def_u}) is generic.

The procedure is only slightly more complicated for higher poles at finite points.

\begin{lemma}
\label{lemma_5}
Assume that $\Omega$ has only simple poles.
Let $q \in {\mathbb K}[z_n]$ be an irreducible polynomial appearing in the denominator of the partial fraction
decomposition of the $\hat{\varphi}_{j}$'s at worst to the power $o$.
A gauge transformation with the seed
\bq
 f_j\left(z_n\right) & = & \frac{1}{q^{o-1}} \sum\limits_{k=0}^{\mathrm{deg}(q)-1} c_{j,k} \; z_n^k,
 \;\;\;\;\;\;\;\;\;
 c_{j,k} \; \in \; {\mathbb K}.
\eq
reduces the order, provided the linear system obtained from 
the condition that in the partial fraction decomposition of 
\bq
\label{eq_seed_finite_point}
 \hat{\varphi}_{j}
 +
 \left( \delta_{j k} \partial_{z_n} + \Omega_{jk} \right) f_k
\eq
terms of the form $z_n^k/q^o$ are absent (with $0 \le k \le \mathrm{deg}(q)-1$)
yield a solution for the $(\nu \cdot \mathrm{deg}(q))$ coefficients $c_{j,k}$.
Furthermore, this gauge transformation does not introduce higher poles elsewhere.
\end{lemma}
\begin{proof}
The requirement that in the partial fraction decomposition of eq.~(\ref{eq_seed_finite_point})
terms of the form $z_n^k/q^o$ are absent 
defines a linear system with $(\nu \cdot \mathrm{deg}(q))$
unknowns and equations.
By assumption, this linear system of equations has a solution for the coefficients $c_{j,k}$.
Having determined the coefficients $c_{j,k}$, we define $\hat{\varphi}_{j}'$ by
\bq
\label{reduction_finite_point}
 \hat{\varphi}_{j}'
 & = &
 \hat{\varphi}_{j}
 +
 \left( \delta_{j k} \partial_{z_n} + \Omega_{jk} \right) f_k
\eq
This reduces the highest power of $q$ in the denominator by one and repeating this procedure we may lower it to one.
As above, the procedure may introduce new simple poles elsewhere (through $\Omega$),
but it will not introduce new higher poles elsewhere (since we assumed that $\Omega$ has only simple poles).
\end{proof}

Readers familiar with integration-by-parts identities in the context of Feynman integrals \cite{Tkachov:1981wb,Chetyrkin:1981qh}
will certainly recognise the analogy: This is a variant of integration-by-parts reduction.
However, we should stress that contrary to the case of Feynman integrals, the size of the involved linear system
is rather modest, it is given by
\bq
 \dim H^{({\bf i})}_\omega \cdot \mathrm{deg}(q),
\eq
where $\dim H^{({\bf i})}_\omega$ corresponds to the number of master integrals at this stage and
$\mathrm{deg}(q)$ gives the degree of one irreducible polynomial in the denominator in the variable $z_i$.

In applications towards Feynman integrals, the following case is of particular interest:
\begin{lemma}
\label{lemma_6}
Assume that $\Omega$ is of the form $\Omega=\gamma \tilde{\Omega}$, where $\tilde{\Omega}$ 
has only simple poles and is independent of $\gamma$.
Assume further that $\gamma$ is infinitesimal.
Then the linear systems appearing in the lemmata~\ref{lemma_4} and~\ref{lemma_5} have a solution.
\end{lemma}
\begin{proof}
In this case the determinant of the linear system is given by
\bq
 \det\left( 1 + \gamma X\right),
\eq
with a matrix $X$ being independent of $\gamma$.
This determinant is non-zero, hence the linear system has a solution.
\end{proof}

% -----------------------------------------------------------------------------

\section{The intersection number of univariate vector-valued one-forms with only simple poles}
\label{sect:forumula}

In this section we investigate the intersection of the coefficients
$\langle \varphi^{({\bf n})}_{L,j}|$ and $| \varphi^{({\bf n})}_{R,j}\rangle$.
We may assume that the coefficients have only simple poles.
In this case we may evaluate the intersection number with the help of a global residue,
which may be computed without introducing algebraic extensions.

Let us shortly summarise what we achieved so far:
In order to compute the intersection number
\bq
 \left\langle \varphi_L \right. \left| \varphi_R \right\rangle
 \;\;\;\;\;\;\;\;\;
 \mbox{for}
 \;\;\;\;\;\;
 \left\langle \varphi_L \right| \in H^n_\omega,
 \;\;\;\;\;\;
 \left| \varphi_R \right\rangle \in \left( H^n_\omega \right)^\ast,
\eq
we expand $\langle \varphi_L |$ in the basis of $H^{({\bf n})}_\omega$
and $| \varphi_R \rangle$ in the dual basis of $(H^{({\bf n})}_\omega)^\ast$
\bq
\label{basis_expansion}
 \left\langle \varphi_L \right|
 \; = \;
 \sum\limits_{j=1}^{\nu_{\bf n-1}}
 \left\langle \varphi^{({\bf n})}_{L,j} \right| \wedge \left\langle e^{({\bf n-1})}_j \right|,
 \;\;\;
 & &
 \;\;\;
 \left| \varphi_R \right\rangle
 \; = \;
 \sum\limits_{j=1}^{\nu_{\bf n-1}}
 \left| d^{({\bf n-1})}_j \right\rangle
 \wedge
 \left| \varphi^{({\bf n})}_{R,j} \right\rangle.
\eq
Due to
\bq
 \left\langle e^{({\bf n-1})}_j \left| d^{({\bf n-1})}_k \right. \right\rangle
 & = &
 \delta_{j k}
\eq
the intersection number becomes
\bq
\label{intersection_coeffs}
 \left\langle \varphi_L \right. \left| \varphi_R \right\rangle
 & = &
 \sum\limits_{j=1}^{\nu_{\bf n-1}}
 \left\langle \varphi^{({\bf n})}_{L,j} \left| \varphi^{({\bf n})}_{R,j} \right. \right\rangle.
\eq
Due to the results of section~\ref{sect:reduction}
we may assume that $\langle \varphi^{({\bf n})}_{L,j} |$ and $| \varphi^{({\bf n})}_{R,j} \rangle$
have only simple poles in $z_n$.
It remains to compute the right-hand side of eq.~(\ref{intersection_coeffs}).

The algorithm of \cite{Mizera:2019gea,Frellesvig:2019uqt} computes the right-hand side of 
eq.~(\ref{intersection_coeffs}) as a sum over the residues at the singular points of $\Omega^{({\bf n})}$.
This requires local solutions $\hat{\psi}^{({\bf n})}_{L,i}$ or $\hat{\psi}^{({\bf n})}_{R,k}$ of
\bq
\label{local_solutions}
 \hat{\psi}^{({\bf n})}_{L,i} \left( \overleftarrow{\partial}_{z_n} \delta_{i j} + \Omega^{({\bf n})}_{i j} \right)
 \; = \; 
 \hat{\varphi}^{({\bf n})}_{L,j}
 & \mbox{or} &
 \left( \partial_{z_n} \delta_{j k} - \Omega^{({\bf n})}_{j k} \right)
 \hat{\psi}^{({\bf n})}_{R,k}
 \; = \; 
 \hat{\varphi}^{({\bf n})}_{R,j}.
\eq
In general, the singular points of $\Omega^{({\bf n})}$ are given by roots.
It is at this stage where the algorithm of \cite{Mizera:2019gea,Frellesvig:2019uqt} introduces algebraic
extensions.

On the other hand, it is known (even in the multi-variate case) \cite{Mizera:2017rqa} that in the case where
all polynomials $p_i$ in eq.~(\ref{def_polynomials}) are linear in the variables $z_j$ (i.e. each polynomial
defines a hyperplane)
and where $\langle \varphi_L |$ and $| \varphi_R \rangle$ have at most a simple pole along the divisor $D$,
the left-hand side can be evaluated as a sum over the residues at the critical points of $\omega$.
The critical points of $\omega$ are the points $(z_1,\dots,z_n) \in {\mathbb C}^n$ where
\bq 
 \omega & = & 0.
\eq
This sum does not involve local solutions of eq.~(\ref{local_solutions}).
It is a global residue and can be evaluated without knowing the positions of the critical points, along the lines of
ref.~\cite{Cattani:2005,Sogaard:2015dba}.

We would like to get around the restriction that all polynomials $p_i$ in eq.~(\ref{def_polynomials}) define
hyperplanes.
Our aim is to evaluate the right-hand side as a global residue over a suitable defined set of ``critical points''.

Let us first discuss the simplest case $n=1$.
Since $\dim H^{({\bf 0})}_\omega = 1$ this is a ``scalar'' case.
We write
\bq
 \omega \; = \; \omega_1 dz_1,
 & &
 \omega_1 \; = \; \frac{P}{Q},
 \;\;\;\;\;\;\;\;\;
 P, Q \; \in \; {\mathbb K}\left[z_1\right],
 \;\;\;\;\;\;\;\;\;
 \gcd\left(P,Q\right) \; = \; 1.
\eq
Let ${\mathcal C}_1$ be the set of critical points of $\omega$, e.g.
\bq
 {\mathcal C}_1 & = & 
 \left\{ \; z_1 \in {\mathbb C} \; | \; P\left(z_1\right) = 0 \; \right\}.
\eq
Closely related to ${\mathcal C}_1$ is the ideal $I_1 \subseteq {\mathbb K}[z_1]$ generated by
\bq
 I_1 & = & 
\left\langle P \right\rangle.
\eq
We have ${\mathcal C}_1 = V(I_1)$, where $V(I)$ denotes the algebraic variety corresponding to the ideal $I$.
$I_1$ is a principal ideal, and $P$ is automatically a Gr\"obner basis for $I_1$.
In the case where $\langle \varphi_L |$ and $| \varphi_R \rangle$ have at most only simple poles in $z_1$ 
the intersection number is given by the global residue
\bq
\label{global_residue_scalar_case}
 \left\langle \varphi_L \right. \left| \varphi_R \right\rangle
 & = &
 - \mathrm{res}_{\langle P \rangle}\left( Q \; \hat{\varphi}_{L} \hat{\varphi}_{R} \right)
 \; = \;
 - \mathrm{res}_{\langle P \rangle}\left( Q \; \hat{\varphi}_{L,1} \hat{\varphi}_{R,1} \right).
\eq
The notation $\mathrm{res}_{\langle P \rangle}(f)$ denotes the global residue of $f$ with respect to $\langle P \rangle$,
i.e. the sum of all local residues of $f$ at the points $z_1$ where $P(z_1)=0$.
As $P$ is a polynomial, this set of points is finite.

Since $\nu_{\bf 0} = \dim H^{({\bf 0})}_\omega = 1$ the expansions in eq.~(\ref{basis_expansion}) are trivial
and we have (with $\langle e^{({\bf 0})}_1 | = | d^{({\bf 0})}_1 \rangle = 1$)
\bq
 \hat{\varphi}_{L} \; = \; \hat{\varphi}_{L,1},
 & &
 \hat{\varphi}_{R} \; = \; \hat{\varphi}_{R,1}.
\eq
We now have to generalise eq.~(\ref{global_residue_scalar_case}) from the scalar case to the vectorial case.
We may think of $\omega_1$ as a $1 \times 1$-matrix. The critical points ${\mathcal C}_1$ are the points,
where $\omega_1$ has not full rank.
Let us now consider the $(\nu_{\bf n-1} \times \nu_{\bf n-1})$-matrix $\Omega^{({\bf n})}$.
We write
\bq
\label{def_P_Q}
 \det\left(\Omega^{({\bf n})}\right) & = & \frac{P}{Q},
 \;\;\;\;\;\;\;\;\;
 P, Q \; \in \; {\mathbb K}\left[z_n\right],
 \;\;\;\;\;\;\;\;\;
 \gcd\left(P,Q\right) \; = \; 1.
\eq
We define ${\mathcal C}_n$ as the set of points, where $\Omega^{({\bf n})}$ does not have full rank, i.e.
\bq
 {\mathcal C}_n & = & 
 \left\{ \; z_n \in {\mathbb C} \; | \; P\left(z_n\right) = 0 \; \right\}.
\eq
Similarly, we define $I_n$ as the ideal in ${\mathbb K}[z_n]$ generated by $P$:
\bq
 I_n & = & 
\left\langle P \right\rangle.
\eq
Again we have ${\mathcal C}_n = V(I_n)$.
$I_n$ is a principal ideal, and $P$ is automatically a Gr\"obner basis for $I_n$.
Finally, we denote by $\mathrm{adj} \; \Omega^{({\bf n})}$ the adjoint matrix of $\Omega^{({\bf n})}$.
This matrix satisfies
\bq
 \Omega^{({\bf n})} \cdot \left( \mathrm{adj} \; \Omega^{({\bf n})} \right)
 \; = \;
 \left( \mathrm{adj} \; \Omega^{({\bf n})} \right) \cdot \Omega^{({\bf n})}
 \; = \;
 \det\left(\Omega^{({\bf n})}\right) \cdot {\bf 1}.
\eq
We may now state the generalisation of eq.~(\ref{global_residue_scalar_case}) to the vectorial case:
\begin{theorem}
\label{theorem_7}
Assume that $\Omega^{({\bf n})}$, $\langle \varphi^{({\bf n})}_{L,j} |$ and $| \varphi^{({\bf n})}_{R,j} \rangle$
have at most only simple poles in $z_n$.
Then
\bq
\label{global_residue_vector_case}
 \left\langle \varphi_L \right| \left. \varphi_R \right\rangle
 & = &
 - \mathrm{res}_{\langle P \rangle}\left( Q \; \hat{\varphi}_{L,i} \left( \mathrm{adj} \; \Omega^{({\bf n})} \right)_{ i j} \hat{\varphi}_{R,j} \right),
\eq
with $Q$ defined as in eq.~(\ref{def_P_Q}).
\end{theorem}
\begin{proof}
We start from eq.~(\ref{basis_expansion})
\bq
 \left\langle \varphi_L \right|
 \; = \;
 \sum\limits_{j=1}^{\nu_{\bf n-1}}
 \left\langle \varphi^{({\bf n})}_{L,j} \right| \wedge \left\langle e^{({\bf n-1})}_j \right|,
 \;\;\;
 & &
 \;\;\;
 \left| \varphi_R \right\rangle
 \; = \;
 \sum\limits_{j=1}^{\nu_{\bf n-1}}
 \left| d^{({\bf n-1})}_j \right\rangle
 \wedge
 \left| \varphi^{({\bf n})}_{R,j} \right\rangle,
\eq
and we assume that the coefficients $\varphi^{({\bf n})}_{L,j}$ and $\varphi^{({\bf n})}_{R,j}$
have only simple poles in the variable $z_n$.
The intersection number is then given by
\bq
\label{to_be_shown}
 \left\langle \varphi_L \right. \left| \varphi_R \right\rangle
 & = &
 \sum\limits_{z_0 \in  {\mathcal S}_{\bf n}} 
 \sum\limits_{j=1}^{\nu_{\bf n-1}}
 \mathop{\mathrm{res}}_{z_n=z_0}
 \left(\hat{\psi}^{({\bf n})}_{L,j}
 \left| \varphi^{({\bf n})}_{R,j} \right\rangle \right)
 \; = \;
 -
 \sum\limits_{z_0 \in  {\mathcal S}_{\bf n}} 
 \sum\limits_{j=1}^{\nu_{\bf n-1}}
 \mathop{\mathrm{res}}_{z_n=z_0} \left(
 \left\langle \varphi^{({\bf n})}_{L,j} \right|
 \hat{\psi}^{({\bf n})}_{R,j}
 \right).
\eq
We have to show that eq.~(\ref{global_residue_vector_case})
agrees with eq.~(\ref{to_be_shown}).
Since $\varphi^{({\bf n})}_{L,j}$ and $\varphi^{({\bf n})}_{R,j}$ have only simple poles
(and $\Omega^{({\bf n})}$ has only simple poles as well),
$\hat{\psi}^{({\bf n})}_{L,j}$ and $\hat{\psi}^{({\bf n})}_{R,j}$ are given locally around $z_n=z_0$ 
by
\bq
 \hat{\psi}^{({\bf n})}_{L,j}
 & = &
 \hat{\varphi}^{({\bf n}),(-1)}_{L,i}
 \left( \Omega^{({\bf n}),(-1)} \right)^{-1}_{ij}
 + {\mathcal O}\left(z_n-z_0\right),
 \nonumber \\
 \hat{\psi}^{({\bf n})}_{R,j}
 & = &
 -
 \left( \Omega^{({\bf n}),(-1)} \right)^{-1}_{jk}
 \hat{\varphi}^{({\bf n}),(-1)}_{R,k}
 + {\mathcal O}\left(z_n-z_0\right),
\eq
where the superscript $({\bf n}),(-1)$ denotes the residue in an expansion around
$z_n=z_0$.
Only the constant part of $\hat{\psi}^{({\bf n})}_{L,j}$ and $\hat{\psi}^{({\bf n})}_{R,j}$
with respect to the variable $z_n$ is relevant to eq.~(\ref{to_be_shown}).
We may replace $\hat{\psi}^{({\bf n})}_{L,j}$ and $\hat{\psi}^{({\bf n})}_{R,j}$ by
\bq
 \hat{\psi}^{({\bf n})}_{L,j}
 \; \rightarrow \;
 \hat{\varphi}^{({\bf n})}_{L,i}
 \left( \Omega^{({\bf n})} \right)^{-1}_{ij},
 \nonumber \\
 \hat{\psi}^{({\bf n})}_{R,j}
 \; \rightarrow  \;
 -
 \left( \Omega^{({\bf n})} \right)^{-1}_{jk}
 \hat{\varphi}^{({\bf n})}_{R,k}
\eq
and eq.~(\ref{to_be_shown}) becomes (with $\det \Omega^{({\bf n})} =P/Q$)
\bq
 \left\langle \varphi_L \right. \left| \varphi_R \right\rangle
 & = &
 \sum\limits_{z_0 \in  {\mathcal S}_{\bf n}} 
 \sum\limits_{i,j=1}^{\nu_{\bf n-1}}
 \mathop{\mathrm{res}}_{z_n=z_0}
 \left(
  \hat{\varphi}^{({\bf n})}_{L,i}
 \left( \Omega^{({\bf n})} \right)^{-1}_{ij}
 \hat{\varphi}^{({\bf n})}_{R,j} \; dz_n \right) 
 \nonumber \\
 & = &
 \sum\limits_{z_0 \in  {\mathcal S}_{\bf n}} 
 \sum\limits_{i,j=1}^{\nu_{\bf n-1}}
 \mathop{\mathrm{res}}_{z_n=z_0}
 \left(
  \frac{Q}{P}
  \hat{\varphi}^{({\bf n})}_{L,i}
 \left( \mathrm{adj} \; \Omega^{({\bf n})}  \right)_{ij}
 \hat{\varphi}^{({\bf n})}_{R,j} dz_n \right) 
 \nonumber \\
 & = &
 \frac{1}{2\pi i}
 \int\limits_{{\mathcal C}}
  dz_n
 \sum\limits_{i,j=1}^{\nu_{\bf n-1}}
  \frac{Q \hat{\varphi}^{({\bf n})}_{L,i} \left( \mathrm{adj} \; \Omega^{({\bf n})}  \right)_{ij} \hat{\varphi}^{({\bf n})}_{R,j}}{P},
\eq
where the contour ${\mathcal C}$ consists of small counter-clockwise circles around all singular points $z_0 \in {\mathcal S}_n$.
We may deform this contour such that the contour goes from a singular point to infinity, comes
back from infinity to half-encircle the next singular point counter-clockwise, goes back to infinity etc..
This contour encloses all critical points $P=0$ clockwise.
\begin{figure}
\begin{center}
\includegraphics[scale=1.0]{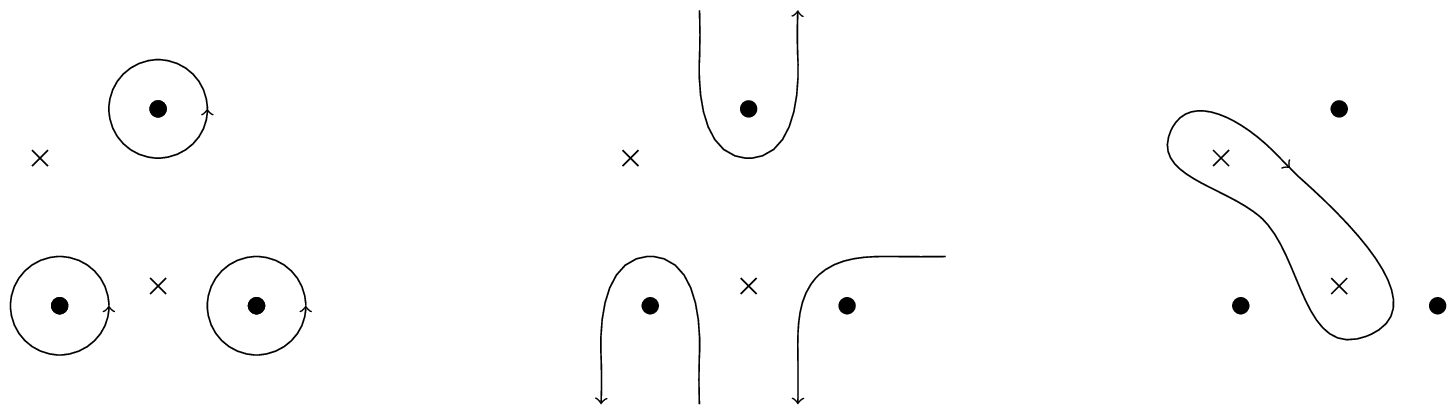}
\end{center}
\caption{
The deformation of the contour.
Singular points are drawn with a dot, critical points by a cross.
Left: The original integration contour encircles all singular points anti-clockwise.
Middle: The integration contour is deformed to infinity.
Right: The integration contour is further deformed to enclose all critical points clockwise.
}
\label{fig_deformation}
\end{figure}
This is shown in fig.~\ref{fig_deformation}.
Localising the integral on $P=0$ gives eq.~(\ref{global_residue_vector_case}), including the minus sign due to the clockwise
orientation.
This completes the proof.
\end{proof}

Eq.~(\ref{global_residue_vector_case}) is the main result of this paper.
The right-hand side is again a global residue (in one variable $z_n$) and can be computed without
introducing algebraic extensions.

% -----------------------------------------------------------------------------

\section{Computation of the global residue}
\label{sect:global_residue}

In this section we review how to compute a global residue in one variable without introducing algebraic
extensions.
The method is an adoption of ref.~\cite{Cattani:2005,Sogaard:2015dba} to the univariate case,
for the underlying mathematics we refer to ref.~\cite{Griffiths:book}.

Let $P \in {\mathbb K}[z]$ be a polynomial in $z$ and let $f(z)$ be a rational function in $z$ with coefficients in ${\mathbb K}$.
We write
\bq
 f \; = \; \frac{P_f}{Q_f},
 \;\;\;\;\;\;\;\;\;
 P_f, Q_f \; \in \; {\mathbb K}\left[z\right],
 \;\;\;\;\;\;\;\;\;
 \gcd\left(P_f, Q_f\right) \; = \; 1.
\eq
We would like to compute the global residue
\bq
 \mathrm{res}_{\langle P \rangle}\left( f \right).
\eq
We denote the ideal generated by $P$ by $I=\langle P \rangle$.
We may assume that $P$ and $Q_f$ have no common zero, e.g. $f$ is not singular on the critical points
$V(I)$.
By Hilbert's Nullstellensatz there exist polynomials $\tilde{P}$ and $\tilde{Q}_f$ in ${\mathbb K}[z]$ such that
\bq
\label{Nullstellensatz}
 \tilde{P} P + \tilde{Q}_f Q_f & = & 1.
\eq
$\tilde{Q}_f$ is called the polynomial inverse of $Q_f$ with respect to the ideal $I$.
The polynomials $\tilde{P}$ and $\tilde{Q}_f$ can be computed with the extended Euclidean algorithm.

With the polynomial inverse at hand 
we have
\bq
\label{eq_polynomial_inverse}
 \mathrm{res}_{\langle P \rangle}\left( f \right)
 & = &
 \mathrm{res}_{\langle P \rangle}\left( P_f \tilde{Q}_f \right).
\eq
Eq.~(\ref{eq_polynomial_inverse}) allows us to replace a calculation with rational functions by a calculation with polynomials.

Let us now consider the vector space
\bq
 {\mathbb K}\left[z\right] / I.
\eq
This vector space has dimension $\nu=\deg P$.
A monomial basis for this vector space is given by
\bq
 v_j \; = \; z^{j-1},
 \;\;\;\;\;\;\;\;\;
 1 \; \le \; j \; \le \; \nu.
\eq
By polynomial division with remainder we may write
\bq
 P_f \tilde{Q}_f
 & = &
 \sum\limits_{j=1}^\nu
 a_j v_j
 \mod \; I.
\eq
The global residue defines a non-degenerate symmetric inner product on ${\mathbb K}[z] / I$:
\bq
 \left( P_1, P_2 \right)
 & = &
 \mathrm{res}_{\langle P \rangle}\left( P_1 \cdot P_2 \right),
 \;\;\;\;\;\;\;\;\;
 P_1, P_2 \; \in \; {\mathbb K}\left[z\right] / I.
\eq
Let $w_j$ be the dual basis to $v_j$ with respect to this inner product, e.g.
\bq
 \left( v_i, w_j \right)
 & = &
 \delta_{i j}.
\eq
We write $1$ as a linear combination of the dual basis
\bq
 1 & = &
 \sum\limits_{j=1}^\nu
 b_j w_j
 \mod \; I.
\eq
Then
\bq
 \mathrm{res}_{\langle P \rangle}\left( f \right)
 \; = \;
 \mathrm{res}_{\langle P \rangle}\left( P_f \tilde{Q}_f \right)
 \; = \;
 \sum\limits_{j=1}^\nu
 a_j b_j.
\eq
Thus it remains to give an algorithm for the computation of the dual basis $w_j$.
This can be done with the Bezoutian matrix, which in our case is just a $1 \times 1$-matrix.
We define
\bq
 B\left(z,y\right)
 & = &
 \frac{P\left(z\right)-P\left(y\right)}{z-y}.
\eq
$B(z,y)$ is a polynomial of degree $(\nu-1)$ in $z$ and $y$.
One expands $B(z,y)$ in $y$. The coefficient of $y^{j-1}$ is a polynomial $w_j(z)$ in $z$
and defines the dual basis $w_j$.
For the case at hand this can be done once and for all:
If 
\bq
 P & = & \sum\limits_{j=0}^\nu c_j z^j,
 \;\;\;\;\;\;\;\;\;
 c_j \; \in \; {\mathbb K},
\eq
then
\bq
 w_j
 & = &
 \sum\limits_{k=0}^{\nu-j} c_{k+j} z^k.
\eq
In particular
\bq
 w_{\nu} & = & c_\nu
\eq
and the global residue reduces to
\bq
 \mathrm{res}_{\langle P \rangle}\left( f \right)
 & = &
 \frac{a_\nu}{c_\nu},
\eq
where $a_\nu$ is the coefficient of $z^{\nu-1}$ in the reduction of $P_f \tilde{Q}_f$ modulus $P$.

Let us summarise:
\begin{proposition}
\label{proposition_8}
Let $P, P_f, Q_f \in {\mathbb K}[z]$ with $\gcd(P, Q_f) \; = \; \gcd(P_f, Q_f) \; = \; 1$.
Let $\nu=\deg P$ and let
$\tilde{Q}_f$ be the polynomial inverse of $Q_f$ with respect to the ideal $\langle P \rangle$.
Then
\bq
 \mathrm{res}_{\langle P \rangle}\left( \frac{P_f}{Q_f} \right)
 & = &
 \frac{a_\nu}{c_\nu},
\eq
where $a_\nu$ is the coefficient of $z^{\nu-1}$ in the reduction of $P_f \tilde{Q}_f$ modulus $P$
and $c_\nu$ is the coefficient of $z^{\nu}$ of $P$.
\end{proposition}
\begin{proof}
See for example \cite{Cattani:2005,Sogaard:2015dba}.
The main ideas have been outlined above.
\end{proof}

% -----------------------------------------------------------------------------

\section{Bases for the twisted cohomology groups}
\label{sect:bases}

Within the recursion we need bases for the twisted cohomology groups $H^{({\bf i})}_\omega$
for $0 \le i \le n-1$.
The dimension of the twisted cohomology groups is given by the number of critical points of 
$\omega^{({\bf i})}$ \cite{Lee:2013hzt,Frellesvig:2019uqt}
\bq
 \dim H^{({\bf i})}_\omega & = &
 \# \; \mbox{solutions of } \omega^{({\bf i})} = 0 \;\; \mbox{on} \;\; {\mathbb C}^n - D.
\eq
Usually it is not an issue to find a basis.
For completeness, we give here a systematic algorithm to construct a basis for $H^{({\bf i})}_\omega$
for the case where all critical points are proper and non-degenerate
(although it involves the computation of a multivariate Gr\"obner basis).
We write
\bq
 \omega^{({\bf i})}
 & = &
 \sum\limits_{j=1}^i \omega_j dz_j,
 \;\;\;\;\;\;
 \omega_j \; = \; \frac{P_j}{Q_j},
 \;\;\;\;\;\;
 P_j, Q_j \; \in \; \tilde{\mathbb K}\left[z_1,\dots,z_i\right],
 \;\;\;\;\;\;
 \gcd\left(P_j,Q_j\right) \; = \; 1,
\eq
with $\tilde{\mathbb K} = {\mathbb K}(z_{i+1},\dots,z_n)$.
We consider the ideal
\bq
 I_i \; = \; \left\langle P_1, \dots, P_i \right\rangle
 \; \subset \; 
 \tilde{\mathbb K}\left[z_1,\dots,z_i\right].
\eq
In the case where all critical points are proper and non-degenerate we have
\bq
 \dim H^{({\bf i})}_\omega 
 & = &
 \dim\left( \tilde{\mathbb K}\left[z_1,\dots,z_i\right] / I_i \right).
\eq
Let $G_1, \dots, G_r$ be a Gr\"obner basis of $I_i$ with respect to some term order $<$:
\bq
 I_i & = &
 \left\langle G_1, \dots,G_r \right\rangle.
\eq
A basis for $H^{({\bf i})}_\omega$ is given by all monomials
\bq
 \prod\limits_{k=1}^i z_k^{\nu_k},
 \;\;\;\;\;\;\;\;\;
 \nu_k \; \in \; {\mathbb N}_0
\eq
which are not divisible by any leading term of the Gr\"obner basis:
\bq
 \mathrm{lt}\left(G_j\right)
 & \not | &
 \prod\limits_{k=1}^i z_k^{\nu_k}
 \;\;\;\;\;\;\;\;\; \forall \; 0 \; \le \, j \; \le \; r.
\eq
Here, $\mathrm{lt}$ denotes the leading term of a polynomial with respect to the chosen term order.

% -----------------------------------------------------------------------------

\section{The algorithm}
\label{sect:algorithm}

We may now summarise the algorithm for the computation of intersection numbers for twisted cocycles:
\begin{itemize}
\item[] \textbf{Input:} A cohomology class $\langle \varphi_L | \in H^{({\bf n})}_\omega$,
a dual cohomology class $| \varphi_R \rangle \in (H^{({\bf n})}_\omega)^\ast$
and a list of bases $\langle e^{({\bf i})}_j |$ of $H^{({\bf i})}_\omega$
for $0 \le i \le (n-1)$.
\item[] \textbf{Output:} The intersection number $\langle \varphi_L | \varphi_R \rangle$.
\end{itemize}
\begin{itemize}
\item[1)] Recursion stop: If $n=0$ return $\varphi_L \varphi_R$.
\item[2)] Computations with $(n-1)$ variables:
Compute the dual basis $| d^{({\bf n-1})}_j \rangle$ of $(H^{({\bf n-1})}_\omega)^\ast$
and the matrix $\Omega^{({\bf n})}$.
Expand $\langle \varphi_L |$ in the basis $\langle e^{({\bf n-1})}_j |$ of $H^{({\bf n-1})}_\omega$
\bq
 \left\langle \varphi_L \right|
 & = &
 \sum\limits_{j=1}^{\nu_{\bf n-1}}
 \left\langle \varphi_{L,j} \right| \wedge \left\langle e^{({\bf n-1})}_j \right|
\eq
and expand $| \varphi_R \rangle$ in the dual basis $| d^{({\bf n-1})}_j \rangle$ of $(H^{({\bf n-1})}_\omega)^\ast$
\bq
 \left| \varphi_R \right\rangle
 & = &
 \sum\limits_{j=1}^{\nu_{\bf n-1}}
 \left| d^{({\bf n-1})}_j \right\rangle
 \wedge
 \left| \varphi_{R,j} \right\rangle.
\eq
\item[3)] Reduction to simple poles:
Reduce the coefficient vector $\varphi_{L,j}$ to an equivalent vector $\varphi_{L,j}'$ with only simple poles in the variable $z_n$.
Similarly, reduce the coefficient vector $\varphi_{R,j}$ to an equivalent vector $\varphi_{R,j}'$ with only simple poles in the variable $z_n$.
\item[4)] Global residue:
Define the polynomials $P,Q \in {\mathbb K}\left[z_n\right]$ by
\bq
 \det\left(\Omega^{({\bf n})}\right) & = & \frac{P}{Q},
 \;\;\;\;\;\;\;\;\;
 \gcd\left(P,Q\right) \; = \; 1.
\eq
Return the univariate global residue
\bq
 \left\langle \varphi_L \right| \left. \varphi_R \right\rangle
 & = &
 - \mathrm{res}_{\langle P \rangle}\left( Q \; \hat{\varphi}_{L,i}' \left( \mathrm{adj} \; \Omega^{({\bf n})} \right)_{ i j} \hat{\varphi}_{R,j}' \right).
\eq
\end{itemize}

% -----------------------------------------------------------------------------

\section{Examples}
\label{sect:examples}

\subsection{An univariate example}

We start with a univariate example ($n=1$). Let
\bq
 p_1 \; = \; z_1,
 & &
 p_2 \; = \; z_1^6 + z_1^5 + z_1^4 + z_1^3 + z_1^2 + z_1 + 1.
\eq
$p_2$ is the $7$-th cyclotomic polynomial with roots $\exp(2\pi i j/7)$, where $j \in \{1,\dots,6\}$.
Set
\bq
 u
 & = &
 \left( p_1 p_2 \right)^\gamma.
\eq
The differential one form $\omega$ is then given by
\bq
 \omega
 & = &
 \gamma \frac{7 z_1^6 + 6 z_1^5 + 5 z_1^4 + 4 z_1^3 + 3 z_1^2 + 2 z_1 + 1}{p_1 p_2} dz_1.
\eq
A basis $\hat{e}^{(1)}_j$ for $H^1_\omega$ is given by
\bq
 \left( 1, z_1, z_1^2, z_1^3, z_1^4, z_1^5 \right).
\eq
Let us now consider
\bq
 \varphi_L \; = \; \frac{dz_1}{z_1^2},
 & &
 \varphi_R \; = \; dz_1.
\eq
$\varphi_L$ has a double pole at $z_1=0$, $\varphi_R$ has a double pole at $z_1=\infty$.
We have
\begin{align}
 \varphi_L
 & = \; 
 \varphi_L'
 + \nabla_\omega \left( - \frac{1}{\left(1-\gamma\right) z_1} \right),
 & 
 \varphi_L'
 & = \;
 \frac{\gamma}{\left(1-\gamma\right)} \frac{\left(6z_1^5+5z_1^4+4z_1^3+3z_1^2+2z_1+1\right)}{p_1 p_2} dz_1,
 \nonumber \\
 \varphi_R
 & = \; 
 \varphi_R' + \nabla_{-\omega} \left( \frac{z_1 }{1-7\gamma} \right),
 & 
 \varphi_R'
 & = \; 
 - \frac{\gamma}{\left(1-7\gamma\right)} \frac{\left(z_1^5+2z_1^4+3z_1^3+4z_1^2+5z_1+6\right)}{p_2} dz_1.
\end{align}
$\varphi_L'$ and $\varphi_R'$ have only simple poles.
Thus
\bq
 \left\langle \varphi_L \right| \left. \varphi_R \right\rangle
 \; = \;
 \left\langle \varphi_L' \right| \left. \varphi_R' \right\rangle
 & = & 
 \frac{6 \gamma}{\left(1-\gamma\right)\left(1-7\gamma\right)}.
\eq
The results from the algorithm presented here and the algorithm of \cite{Mizera:2019gea,Frellesvig:2019uqt} agree.
However the algorithm presented here does not require the introduction of an algebraic extension 
(in this case the root $r_7=\exp(2\pi i/7)$) in intermediate steps of the calculation.

\subsection{An example with an elliptic curve}
\label{sect:example_elliptic_curve_1}

Let us now consider two variables ($n=2$). 
Let
\bq
 p_1 \; = \; z_1,
 \;\;\;
 p_2 \; = \; z_2,
 \;\;\;
 p_3 \; = \; z_2^2 - 4 z_1^3 + 11 z_1 - 7.
\eq
The cubic equation $4 z_1^3 - 11 z_1 + 7 = 0$ has the roots $z_1^{(0)}=1$, $z_1^{(\pm)}=-1/2\pm \sqrt{2}$.
We set
\bq
 u
 & = &
 \left( p_1 p_2 p_3 \right)^\gamma.
\eq
The differential one-form $\omega$ reads
\bq
 \omega
 & = &
 \gamma \frac{z_2^2 - 16 z_1^3 + 22 z_1 - 7}{p_1 p_3} dz_1
 +
 \gamma \frac{3 z_2^2 - 4 z_1^3 + 11 z_1 - 7}{p_2 p_3} dz_2.
\eq
A basis $\hat{e}^{(2)}_j$ for $H^2_\omega$ is given by
\bq
 \left( 1, z_1, z_2, z_1 z_2 , z_1^2, z_1^2 z_2 \right).
\eq
Let us now consider
\bq
 \varphi_L \; = \; \frac{1}{p_3} dz_2 \wedge dz_1,
 & &
 \varphi_R \; = \; \frac{z_1}{p_3} dz_1 \wedge dz_2.
\eq
We have
\bq
\label{result_example_2}
 \left\langle \varphi_L \right| \left. \varphi_R \right\rangle
 & = &
 \frac{1}{4 \left(1-11\gamma\right) \gamma}.
\eq
With the algorithm presented here, the intersection number is computed
without introducing any algebraic extensions.
We have verified the result with the algorithm of ref.~\cite{Mizera:2019gea,Frellesvig:2019uqt}, which
introduces in intermediate stages algebraic extensions.
In addition, it is advantageous to use for the algorithm of ref.~\cite{Mizera:2019gea,Frellesvig:2019uqt}
the order $(z_2,z_1)$ instead of $(z_1,z_2)$. 
With the order $(z_2,z_1)$ the algorithm of ref.~\cite{Mizera:2019gea,Frellesvig:2019uqt} requires only
the roots $z_1^{(\pm)}$, with the order $(z_1,z_2)$ one would need cubic roots.

With the algorithm presented here, the intersection number can be computed in any order.
Of course, the order may influence the performance.
The dimensions of the ``inner'' cohomology groups $H^{({\bf i})}_\omega$ depends on the chosen order.
As a general rule, it is advantageous to choose an order such that the dimensions
of the ``inner'' cohomology groups are minimised.
In this sense the order $(z_2,z_1)$ is preferred over the order $(z_1,z_2)$, since
\bq
 \mbox{Basis} \; \hat{e}^{(z_2)}_j \; \mbox{of} \; H^{(z_2)}_\omega 
 \; : \;
 \left( 1, z_2 \right),
 & &
 \mbox{basis} \; \hat{e}^{(z_1)}_j \; \mbox{of} \; H^{(z_1)}_\omega 
 \; : \;
 \left( 1, z_1, z_1^2 \right).
\eq
Here we used the notation that $H^{(z_{1/2})}_\omega$ denotes $H^{({\bf 1})}_\omega$ with the
order $(z_1,z_2)$ or $(z_2,z_2)$, respectively.
Analogously, we denote the corresponding connection matrices $\Omega^{({\bf 2})}$
by $\Omega^{(z_1,z_2)}$ for the order $(z_1,z_2)$ and
by $\Omega^{(z_2,z_1)}$ for the order $(z_2,z_1)$.

For the order $(z_1,z_2)$ we have that $\Omega^{(z_1,z_2)}$ is a $3 \times 3$-matrix. The determinant is given by
\bq
\lefteqn{
 \det\left( \Omega^{(z_1,z_2)} \right)
 = } & &
 \\
 & &
 \frac{ \left(2+11\gamma\right)\left(4+11\gamma\right)\left(6+11\gamma\right)z_2^6 
        -231 \gamma \left(33 \gamma^2 + 24 \gamma + 4 \right) z_2^4
        + 2 \gamma^2 \left( 3949 \gamma + 1315 \right) z_2^2
        + 56 \gamma^3
      }{z_2^3 \left(z_2^2-7\right) \left(27 z_2^4 - 378 z_2^2 - 8\right)}.
 \nonumber 
\eq
For the order $(z_2,z_1)$ we have that $\Omega^{(z_2,z_1)}$ is a $2 \times 2$-matrix. The determinant is given by
\bq
\lefteqn{
 \det\left( \Omega^{(z_2,z_1)} \right)
 = } & &
 \\
 & &
 \frac{\left[4\left(3+11\gamma\right)z_1^3 - 11 \left(1+5\gamma\right)z_1 + 14 \gamma \right]
       \left[4\left(6+11\gamma\right)z_1^3 - 11 \left(2+5\gamma\right)z_1 + 14 \gamma \right]}
      {4 z_1^2 \left(z_1-1\right)^2 \left(4 z_1^2 + 4 z_1 -7 \right)^2}.
 \nonumber
\eq
In both cases, the numerator of the determinant is a degree six polynomial 
in the remaining integration variable.
In both cases the determinant has six critical points (defined by the vanishing of the determinant).
This is consistent with
\bq
 \dim H^2_\omega  & = & 6.
\eq
On the other hand, the number of distinct singular points of the determinant 
(defined by the vanishing of the denominator of the determinant)
is given by (not counting multiplicities)
\bq
 \left| {\mathcal S}_{(z_1,z_2)} \right|
 \; = \; 7,
 & &
 \left| {\mathcal S}_{(z_2,z_1)} \right|
 \; = \; 4,
\eq
and has no particular meaning.

Let us give some details on the calculation of the intersection number in eq.~(\ref{result_example_2}).
We choose the order $(z_2,z_1)$, i.e. we integrate out $z_1$ first.
As already noted above, a basis for $H^{(z_2)}_\omega$ is given by 
\bq
 \hat{e}^{(z_2)}_1 \; = \; 1,
 & &
 \hat{e}^{(z_2)}_2 \; = \; z_2.
\eq
We set $\hat{h}^{(z_2)}_1=1$, $\hat{h}^{(z_2)}_2=z_2$, compute the $2 \times 2$-intersection matrix 
and construct the dual basis according to eq.~(\ref{construction_dual_basis}).
The dual basis of $(H^{(z_2)}_\omega)^\ast$ is then
\bq
 \hat{d}^{(z_2)}_1 \; = \; \frac{\left(1-3\gamma\right)\left(1+3\gamma\right)}{2\gamma\left(4z_1^3-11z_1+7\right)},
 & &
 \hat{d}^{(z_2)}_2 \; = \; \frac{3\left(2-3\gamma\right)\left(2+3\gamma\right)}{2\gamma\left(4z_1^3-11z_1+7\right)^2} z_2.
\eq
In this example, $\Omega^{(z_2,z_1)}$ is a diagonal matrix given by
\bq
 \Omega^{(z_2,z_1)}
 & = &
 \left( \begin{array}{cc}
 \frac{12z_1^2-11}{2\left(4z_1^3-11z_1+7\right)} + \frac{44z_1^3-55z_1+14}{2 z_1 \left(4z_1^3-11z_1+7\right)} \gamma & 0 \\
 0 & \frac{12z_1^2-11}{\left(4z_1^3-11z_1+7\right)} + \frac{44z_1^3-55z_1+14}{2 z_1 \left(4z_1^3-11z_1+7\right)} \gamma \\
 \end{array} \right).
 \;\;\;
\eq
We expand $\varphi_L$ in the basis of $H^{(z_2)}_\omega$
\bq
 \varphi_L
 & = & 
 \left( \hat{\varphi}^{(z_2,z_1)}_{L,1} dz_1 \right) \wedge \left( \hat{e}^{(z_2)}_1 dz_2 \right),
 \;\;\;\;\;\;
 \hat{\varphi}^{(z_2,z_1)}_{L,1} \; = \;
 \frac{1+3\gamma}{2\gamma\left(4z_1^3-11z_1+7\right)}.
\eq
The coefficient $\hat{\varphi}^{(z_2,z_1)}_{L,1}$ has already only simple poles in $z_1$, so no additional reduction to simple poles is required.
We expand $\varphi_R$ in the basis of $(H^{(z_2)}_\omega)^\ast$
\bq
 \varphi_R
 & = & 
 \left( \hat{d}^{(z_2)}_1 dz_2 \right) \wedge \left( \hat{\varphi}^{(z_2,z_1)}_{R,1} dz_1 \right),
 \;\;\;\;\;\;
 \hat{\varphi}^{(z_2,z_1)}_{R,1} \; = \;
 - \frac{z_1}{\left(1+3\gamma\right)}.
\eq
The coefficient $\hat{\varphi}^{(z_2,z_1)}_{R,1}$ has a triple pole at infinity. The reduction to simple poles gives
\bq
 \varphi_R'
 & = & 
 \left( \hat{d}^{(z_2)}_1 dz_2 \right) \wedge \left( \hat{\varphi}^{(z_2,z_1)}_{R,1}{}' dz_1 \right),
 \;\;\;\;\;\;
 \hat{\varphi}^{(z_2,z_1)}_{R,1}{}' \; = \;
 \frac{z_1\left(21-22z_1\right)}{\left(1-11\gamma\right)\left(4z_1^3-11z_1+7\right)}.
\eq
The intersection number is then obtained as
\bq
 \left\langle \varphi_L \right| \left. \varphi_R \right\rangle
 & = &
 \left\langle \hat{\varphi}^{(z_2,z_1)}_{L,1} dz_1 \right| \left. \hat{\varphi}^{(z_2,z_1)}_{R,1}{}' dz_1 \right\rangle
 \; = \;
 \frac{1}{4 \left(1-11\gamma\right) \gamma}.
\eq

\subsection{Third example}

As a third example we discuss an example already discussed in \cite{Matsubara-Heo:2019}.
We set
\bq
 p_1 \; = \; z_1,
 \;\;\;
 p_2 \; = \; z_2,
 \;\;\;
 p_3 \; = \; z_1^2 z_2 + z_1 z_2^2 + z_1 + a_4 z_1 z_2 + a_5 z_2,
\eq
where $a_4$ and $a_5$ are two parameters.
$p_3$ is again a genus $1$ curve.
We set
\bq
 u & = &
 p_1^{\frac{1}{2}+\eps}
 p_2^{\frac{1}{2}+\eps}
 p_3^{-\frac{1}{2}}.
\eq
In this case we have $\dim H^2_\omega = 4$ and 
a basis $\hat{e}^{(2)}_j$ for $H^2_\omega$ is given by
\bq
 \left( 
  \frac{1}{z_1 z_2}, 
  \frac{1}{z_1 z_2} \frac{\partial \ln u}{\partial a_5}, 
  \frac{1}{z_1 z_2} \frac{\partial \ln u}{\partial a_4}, 
  \frac{1}{z_1 z_2 u} \frac{\partial^2 u}{\partial a_5^2} \right).
\eq
Let us now consider
\bq
 \varphi_L \; = \; \frac{1}{z_1 z_2} dz_2 \wedge dz_1,
 & &
 \varphi_R \; = \; \frac{1}{z_1 z_2} dz_1 \wedge dz_2.
\eq
We have
\bq
 \left\langle \varphi_L \right| \left. \varphi_R \right\rangle
 & = &
 \frac{32}{1 - 16 \eps^2},
\eq
in agreement with ref.~\cite{Matsubara-Heo:2019}.

% -----------------------------------------------------------------------------

\section{Applications}
\label{sect:applications}

In this section we discuss applications towards Feynman integrals.
We show how information on the system of differential equations for a family of Feynman integrals 
may be obtained from intersection numbers.
The formalism has already been discussed in \cite{Mastrolia:2018uzb,Frellesvig:2019kgj,Frellesvig:2019uqt}.
We may either use the Baikov representation \cite{Baikov:1996iu,Lee:2009dh} or the Lee-Pomeransky \cite{Lee:2013hzt} representation of Feynman integrals.
For concreteness, we focus here on the Baikov representation.
As a pedagogical example we choose the massive sunrise integral.
On the one hand, this example shows that the method works in the multivariate case for higher degree polynomials beyond multiple polylogarithms.
On the other hand with our algorithm at hand 
we are able to clarify a subtlety in the equal mass case first noticed in \cite{Frellesvig:2019kgj}.
We will start by discussing the unequal mass case.
Specialising in section~\ref{sect:equal_mass_sunrise} to the equal mass allows us 
to demonstrate that assumption (\extraassumption) in section~\ref{sect:notation} is required.
We note that the massive sunrise integral in the Lee-Pomeransky representation has been considered in \cite{Mizera:2019vvs}.

In section~\ref{sect:feynman_integral_reduction} we consider Feynman integral reduction. We discuss a non-planar two-loop example relevant to Higgs decay.

\subsection{The Baikov representation}

Our starting point is a $l$-loop $n$-point Feynman integral
\bq
 I_{\nu_1 \dots \nu_{\NV}} & = &
 e^{l \eps \Eulerconstant}
 \left(\mu^2\right)^{\nu-\frac{l D}{2}}
 \int \prod\limits_{r=1}^l \frac{d^Dk_r}{i \pi^{\frac{D}{2}}}
 \prod\limits_{s=1}^{\NV}
 \frac{1}{\left( -q_s^2 + m_s^2 \right)^{\nu_s}},
 \;\;\;\;\;\;
 \nu \; = \; \sum\limits_{s=1}^{\NV} \nu_s,
 \;\;\;\;\;\;
 \nu_s \; \in \; {\mathbb Z}.
\eq
$\Eulerconstant$ is Euler's constant.
Let $p_1, p_2, ..., p_r$ denote the external momenta
and denote by
\bq
 e & = &
 \dim \left\langle p_1, p_2, ..., p_r \right\rangle
\eq
the dimension of the span of the external momenta.
For generic external momenta and $D \ge r-1$ we have $e=r-1$.
We set
\bq
 \NV & = &
 \frac{1}{2} l \left(l+1\right) + e l.
\eq
$\NV$ gives the number of linear independent scalar products involving the loop momenta.
We denote these scalar products by
\bq
 \sigma & = &
 \left( \sigma_1, ..., \sigma_{\NV} \right)
 \; = \; 
 \left( -k_1 \cdot k_1, -k_1 \cdot k_2, ..., -k_{l-1} \cdot k_l, -k_1 \cdot p_1, ..., -k_l \cdot p_e \right).
\eq
We define a $\NV \times \NV$-matrix $C$ and a $\NV$-vector $f$ by
\bq
 -q_s^2 + m_s^2 & = & C_{st} \sigma_t + f_s.
\eq
In order to arrive at the Baikov representation \cite{Baikov:1996iu,Frellesvig:2017aai,Bosma:2017ens,Harley:2017qut}
we change the integration variables to 
the Baikov variables $z_s$:
\bq
 z_s & = & -q_s^2+m_s^2.
\eq
We have
\bq
\label{sigma_to_z}
 \sigma_t
 & = &
 \left( C^{-1} \right)_{t s} \left( z_s - f_s \right).
\eq
The Baikov representation of $I$ is given by
\bq
\label{baikov_representation}
 I_{\nu_1 \dots \nu_{\NV}} & = &
 e^{l \eps \Eulerconstant}
 \left(\mu^2\right)^{\nu-\frac{l D}{2}}
 \frac{\pi^{-\frac{1}{2}\left(\NV-l\right)}}{\prod\limits_{r=1}^l \Gamma\left(\frac{D-e+1-r}{2}\right)}
 \frac{ G\left(p_1,...,p_e\right)^{\frac{-D+e+1}{2}} }{ \det C }
 \nonumber \\
 & &
 \times
 \int\limits_{\mathcal C} d^{\NV}z \;
 G\left(k_1,...,k_l,p_1,...,p_e\right)^{\frac{D-l-e-1}{2}}
 \prod\limits_{s=1}^{\NV} z_s^{-\nu_s},
\eq
where the Gram determinants are defined by
\bq
 G\left(q_1,...,q_n\right) & = &
 \det\left(-q_i \cdot q_j \right).
\eq
$G(k_1,...,k_l,p_1,...,p_e)$ expressed in the variables $z_s$'s through eq.~(\ref{sigma_to_z}) is called the Baikov polynomial:
\bq
 B\left(z_1,...,z_{\NV}\right) & = & G\left(k_1,...,k_l,p_1,...,p_e\right).
\eq
The domain of integration ${\mathcal C}$ is given by \cite{Grozin:2011mt,Mastrolia:2018uzb}
\bq
 {\mathcal C}
 & = &
 {\mathcal C}_1 \cap {\mathcal C}_2 \cap \dots \cap {\mathcal C}_l
\eq
with
\bq
 {\mathcal C}_j
 & = &
 \left\{
 \frac{G\left(k_j,k_{j+1},...,k_l,p_1,...,p_e\right)}{G\left(k_{j+1},...,k_l,p_1,...,p_e\right)} > 0
 \right\}.
\eq

\subsection{The unequal mass sunrise integral}

Let us now specialise to 
\begin{align}
 z_1 & = \; - k_2^2, &
 z_2 & = \; - \left(k_1-p\right)^2,
 & & \nonumber \\
 z_3 & = \; - k_1^2 + m_1^2, &
 z_4 & = \; - \left(k_1-k_2\right)^2 + m_2^2, &
 z_5 & = \; - \left(k_2-p\right)^2 + m_3^2.
\end{align}
This defines the Baikov variables for the sunrise integral.
\begin{figure}
\begin{center}
\includegraphics[scale=1.0]{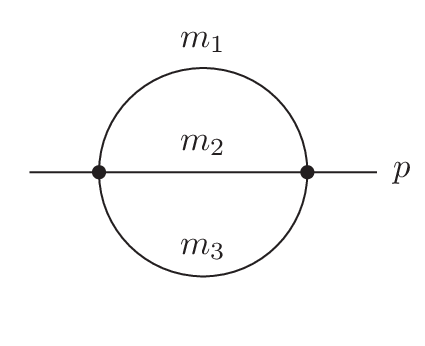}
\end{center}
\caption{
The Feynman graph corresponding to the two-loop sunrise integral.
$p$ denotes the external momentum, $m_1$,$m_2$ and $m_3$ the internal masses.
}
\label{fig_sunrise}
\end{figure}
The corresponding Feynman graph is shown in fig.~\ref{fig_sunrise}.
We have two loops ($l=2$), two external momenta ($r=2$ and $e=1$).
For the dimension of space-time we set $D=2-2\eps$.
We are here only interested in the case where $\nu_1=\nu_2=0$. To shorten the notation, we set
\bq
 S_{\nu_3 \nu_4 \nu_5} & = & I_{0 0 \nu_3 \nu_4 \nu_5}.
\eq
It is well-known that in the unequal mass case there are seven master integrals, which may be taken as
\bq
 \vec{I}
 & = &
 \left(
 S_{011},
 \;\;
 S_{101},
 \;\;
 S_{110},
 \;\;
 S_{111},
 \;\;
 S_{211},
 \;\;
 S_{121},
 \;\;
 S_{112}
 \right)^T.
\eq
We set $\mu=m_3$ and introduce the dimensionless ratios (the notation follows \cite{Bogner:2019lfa})
\bq
 x \; = \; \frac{p^2}{m_3^2}, 
 \;\;\;\;\;\;
 y_1 \; = \; \frac{m_1^2}{m_3^2}, 
 \;\;\;\;\;\;
 y_2 \; = \; \frac{m_2^2}{m_3^2}.
\eq
The derivatives of the master integrals with respect to any of the external variables $(x,y_1,y_2)$ can be expressed again
as a linear combination of the master integrals, for example
\bq
 \frac{\partial}{\partial x} \vec{I} & = & A_x \vec{I},
\eq
where $A_x$ is a $7 \times 7$-matrix.
We are interested in determining the matrix $A_x$.
Traditionally, this is done with the help of integration-by-parts identities.
The use of intersection numbers provides an alternative.
From eq.~(\ref{baikov_representation}) we have
\bq
\label{baikov_representation_sunrise}
 S_{\nu_3 \nu_4 \nu_5}
 & = &
 \frac{\left(-x\right)^\eps}{4 \pi^2 \Gamma\left(-2\eps\right)} 
 \int\limits_{\mathcal C} d^5z \;
 B^{-\eps}
 \frac{1}{z_3^{\nu_3} z_4^{\nu_4} z_5^{\nu_5} B},
\eq
with $B$ being the Baikov polynomial.

\subsubsection{The maximal cut}

In order to determine
\bq
 \left( A_x \right)_{i j},
 & & 4 \; \le \; i,j \; \le \; 7
\eq
it is sufficient to consider the maximal cut.
For the maximal cut we take the three-fold residue $z_3=z_4=z_5=0$.
We set
\bq
 p_1\left(z_1,z_2\right) \; = \; B\left(z_1,z_2,0,0,0\right),
 \;\;\;\;\;\;
 u\left(z_1,z_2\right) \; = \; p_1^{-\eps},
 \;\;\;\;\;\;
 \omega \; = \; d \ln u.
\eq
We further set
\bq 
 \omega_x \; = \; \frac{\partial \ln u}{\partial x},
 \;\;\;\;\;\;\;\;\;
 \omega_{y_1} \; = \; \frac{\partial \ln u}{\partial y_1},
 \;\;\;\;\;\;\;\;\;
 \omega_{y_2} \; = \; \frac{\partial \ln u}{\partial y_2}.
\eq
We have four critical points, consistent with four master integrals on the maximal cut ($S_{111}$, $S_{211}$, $S_{121}$, $S_{112}$).
We set
\bq
\label{def_basis_sunrise_max_cut}
 \hat{e}^{(2)}_{111} & = & \left. \frac{1}{B} \right|_{z_3=z_4=z_5=0} \; = \; \frac{1}{p_1},
 \nonumber \\
 \hat{e}^{(2)}_{211} 
 & = & 
 \left. \left( B^\eps \frac{\partial}{\partial z_3} B^{-\eps-1} \right) \right|_{z_3=z_4=z_5=0} 
 \; = \; 
 - \left(1+\eps\right) \left. \left( \frac{1}{B^2} \frac{\partial B}{\partial z_3} \right) \right|_{z_3=z_4=z_5=0},
 \nonumber \\
 \hat{e}^{(2)}_{121} 
 & = & 
 \left. \left( B^\eps \frac{\partial}{\partial z_4} B^{-\eps-1} \right) \right|_{z_3=z_4=z_5=0} 
 \; = \; 
 - \left(1+\eps\right) \left. \left( \frac{1}{B^2} \frac{\partial B}{\partial z_4} \right) \right|_{z_3=z_4=z_5=0},
 \nonumber \\
 \hat{e}^{(2)}_{112} 
 & = & 
 \left. \left( B^\eps \frac{\partial}{\partial z_5} B^{-\eps-1} \right) \right|_{z_3=z_4=z_5=0} 
 \; = \; 
 - \left(1+\eps\right) \left. \left( \frac{1}{B^2} \frac{\partial B}{\partial z_5} \right) \right|_{z_3=z_4=z_5=0}.
\eq
We denote by $\hat{d}^{(2)}_{111}$, $\hat{d}^{(2)}_{211}$, $\hat{d}^{(2)}_{121}$, $\hat{d}^{(2)}_{112}$
the dual basis.
In order to determine $(A_x)_{4,j}$ we have to consider $dS_{111}/dx$. This corresponds to
\bq
 \hat{\varphi}_L 
 & = &
 \frac{\partial}{\partial x} \hat{e}^{(2)}_{111} + \omega_x \hat{e}^{(2)}_{111} + \frac{\eps}{x} \hat{e}^{(2)}_{111},
\eq
where the last term originates from the prefactor $(-x)^\eps$ in eq.~(\ref{baikov_representation_sunrise}).
The entries $(A_x)_{4,j}$ with $4 \le j \le 7$ are then given by
\bq
 \left(A_x\right)_{4,4}
 \; = \;
 \left\langle \varphi_L | d^{(2)}_{111} \right\rangle,
 \;\;
 \left(A_x\right)_{4,5}
 \; = \;
 \left\langle \varphi_L | d^{(2)}_{211} \right\rangle,
 \;\;
 \left(A_x\right)_{4,6}
 \; = \;
 \left\langle \varphi_L | d^{(2)}_{121} \right\rangle,
 \;\;
 \left(A_x\right)_{4,7}
 \; = \;
 \left\langle \varphi_L | d^{(2)}_{112} \right\rangle,
 \nonumber
\eq
and similar for $(A_x)_{i j}$, $(A_{y_1})_{i j}$ and $(A_{y_2})_{i j}$ for $4 \le i,j \le 7$.
Computing the intersection numbers we find agreement with the known results \cite{Caffo:1998du}.

The polynomial $p_1$ is a degree $3$ polynomial in two variables $(z_1,z_2)$.
The calculation performed here gives an example, where intersection numbers can be applied to Feynman integrals
in the multivariate case and with higher degree polynomials.

\subsection{The equal mass sunrise integral}
\label{sect:equal_mass_sunrise}

Let us consider the equal mass sunrise integral
\bq
 m_1 \; = \; m_2 \; = \; m_3 \; = m \; \neq 0
\eq
and let us focus as before on the maximal cut.
We obtain the correct differential equation on the maximal cut from our results in the unequal mass case by setting 
$m_1=m_2=m_3$ in the end. However, this seems like an overkill.
The equal mass sunrise integral is a simpler Feynman integral with fewer external variables, and we are interested in methods which keep the number of variables to a minimum.

Let us investigate, what happens if we set the masses equal right from the start.
It is well-known that there are three master integrals in the equal mass case.
Due to the additional symmetry related to the masses being equal, the integrands of $S_{011}$, $S_{101}$ and $S_{110}$ integrate to the same functions,
as do the integrands of $S_{211}$, $S_{121}$ and $S_{112}$.
Within the framework of twisted cocycles we deal with integrands and the symmetry is not seen.
Phrased differently, the differential forms are not invariant under permutation of the Baikov variables $(z_3,z_4,z_5)$.
Thus the dimension of the bases will be as in the unequal mass case.

Let us now investigate the maximal cut $z_3=z_4=z_5=0$.
On the maximal cut the Baikov polynomial is given by
\bq
 p_1 & = & \frac{1}{4} \left[ \left(1-x\right)^2 - z_1 z_2 ( z_1 + z_2 + x + 3 ) \right].
\eq
As before we set $u=p_1^{-\eps}$.
There are four critical points, consistent with our expectation that $\dim H^{(\bf 2)}_\omega = 4$.
The critical points are
\bq
\label{critical_points_equal_sunrise}
 z^{(1)} & = & \left(z^{(1)}_1,z^{(1)}_2\right) \; = \; \left(0,0\right),
 \nonumber \\
 z^{(2)} & = & \left(z^{(2)}_1,z^{(2)}_2\right) \; = \; \left(0,-x-3\right),
 \nonumber \\
 z^{(3)} & = & \left(z^{(3)}_1,z^{(3)}_2\right) \; = \; \left(-x-3,0\right),
 \nonumber \\
 z^{(4)} & = & \left(z^{(4)}_1,z^{(4)}_2\right) \; = \; \left(-\frac{x}{3}-1,-\frac{x}{3}-1\right).
\eq
Thus we expect that the equal mass limit of eq.~(\ref{def_basis_sunrise_max_cut}) 
\bq
 \left(
 \hat{e}^{(2)}_{111},
 \hat{e}^{(2)}_{211},
 \hat{e}^{(2)}_{121},
 \hat{e}^{(2)}_{112}
 \right)
\eq
provides a basis $\hat{e}^{(2)}_j$ for $H^2_\omega$.
Let us now naively (i.e. without checking that all assumptions are satisfied)
apply our algorithm (or the algorithm of \cite{Mizera:2019gea,Frellesvig:2019uqt}) 
to compute the intersection matrix.
We expect the intersection matrix to have rank $4$, but we find that the intersection matrix has (erroneously) rank $3$.
This problem was already noted in \cite{Frellesvig:2019kgj}.
However in this publication master integrals (i.e. pairings between twisted cocycles and cycles)
were considered, not twisted cocycles.
After integration we should have two master integrals on the maximal cut and the additional symmetry due to equal masses
brings the number of independent Feynman integrals on the maximal cut down to two independent master integrals.
 
But let us focus on the twisted cocycles. 
A rank $3$ intersection matrix is not correct.
It is instructive to investigate what goes wrong.
We may compare step-by-step the equal mass calculation with the unequal mass calculation, setting in the latter calculation the masses equal
for each comparison.
The problem arises as follows:
The algorithm presented here and the algorithm of \cite{Mizera:2019gea,Frellesvig:2019uqt} both use an recursive approach.
Let's say we first integrate out $z_1$ and then $z_2$.
The matrices $\Omega^{(1)}$ and $\Omega^{(2)}$ are for the case at hand both $1 \times 1$-matrices.
We have
\bq
 \det \Omega^{(1)}
 & = &
 - \frac{\eps z_2 \left( 2 z_1 + x +3 + z_2\right)}{z_2 z_1^2 + z_2 \left(x+3+z_2\right) z_1 - \left(1-x\right)^2},
 \nonumber \\
 \det \Omega^{(2)}
 & = &
 \frac{\left(1-3\eps\right)z_2^3 +\left(1-4\eps\right)\left(x+3\right) z_2^2 - \eps \left(x+3\right)^2 z_2 - 2 \left(1-x\right)^2}{z_2 \left(z_2+4\right) \left[ z_2^2 + 2 \left(x+1\right) z_2 + \left(1-x\right)^2 \right]}.
\eq
In the equal mass case $\det \Omega^{(2)}$ has $3$ critical points (defined as the points $z_2 \in {\mathbb C}$ where $\det \Omega^{(2)}$ vanishes)
and $4$ singular points (defined as the points $z_2 \in {\mathbb C}$ where $\det \Omega^{(2)}$ is singular).
In the unequal mass case $\det \Omega^{(2)}$ has $4$ critical points 
and $5$ singular points.
In the equal mass limit two singular points and one critical point coincide, cancelling a common factor in the numerator and in the denominator
in $\det \Omega^{(2)}$ and leaving as a net result one singular point.
The integration contour separates singular points and critical points.
\begin{figure}
\begin{center}
\includegraphics[scale=1.0]{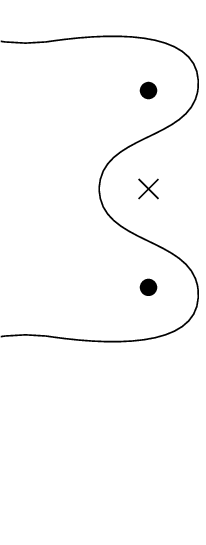}
\hspace*{20mm}
\includegraphics[scale=1.0]{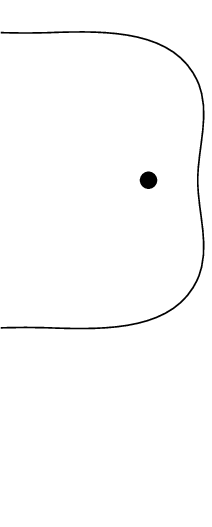}
\hspace*{20mm}
\includegraphics[scale=1.0]{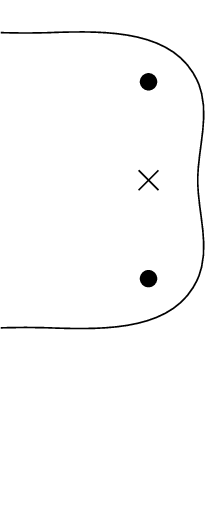}
\end{center}
\caption{
Left: The integration contour in the unequal mass case with two singular points (dots) and one critical point (cross).
Middle: The integration contour in the equal mass case with one singular points.
Right: The one singular point can be considered as the limit where two singular points and one critical point
coincide.
}
\label{fig_contour}
\end{figure}
The situation is shown in fig.~\ref{fig_contour}.
Let's assume we compute the sum of the residues of the critical points.
From fig.~\ref{fig_contour} it is clear that we miss in the equal mass case the contribution from the ``cancelled'' critical point.
This would be o.k., if the contribution from this residue would be zero.
Looking at eq.~(\ref{global_residue_vector_case}),
the two singular points provide two powers in the numerator, 
however $\hat{\varphi}_L$ and $\hat{\varphi}_r$ each are allowed to have a simple pole, cancelling 
the two powers in the numerator and leaving a non-zero residue.

Let us also discuss what happens if we perform a sum of the residues of the singular points along the lines of refs.~\cite{Mizera:2019gea,Frellesvig:2019uqt}.
In the first step we integrate out $z_1$ and sum over the residues in $z_1$ located at the two singular points defined by the
vanishing of the denominator of $\det \Omega^{(1)}$.
We do this for generic $z_2$. For the specific value $z_2=0$ we see that $\Omega^{(1)}$ vanishes and the equation~(\ref{holomorphic_solution}) will
have no solution.
At $z_2=0$ we have a singular fibre.
In the second step we integrate out $z_2$ and sum over 
the residues in $z_2$ located at the four singular points defined by the
vanishing of the denominator of $\det \Omega^{(2)}$.
One of the singular points is $z_2=0$, which a posteriori invalidates the inner integration.

Let us return to the analysis based on critical points.
We see that the assumption (\extraassumption) in section~\ref{sect:notation} is violated: $\det \Omega^{(2)}$ has in the equal mass case only
three critical points, but should have four.
This will happen for the integration order $(z_1,z_2)$ as for the integration order $(z_2,z_1)$.
We see that assumption (\extraassumption) in section~\ref{sect:notation} is a necessary condition.
This is also clear from ref.~\cite{Lee:2013hzt}:
The number of critical points corresponds to the number of independent integration cycles and by duality to the number of independent cocycles.
Having identified the problem, it is easy to find a fix: An inspection of eq.~(\ref{critical_points_equal_sunrise})
shows, that for the integration order $(z_1,z_2)$ (or the integration order $(z_2,z_1)$) two of the four original critical points in $(z_1,z_2)$-space
are in the same fibre.
A coordinate transformation
\bq
 \left( \begin{array}{c}
 z_1 \\
 z_2 \\
 \end{array} \right)
 & = &
 \left( \begin{array}{rr}
 c & s \\
 -s & c \\
 \end{array} \right)
 \left( \begin{array}{c}
 z_1' \\
 z_2' \\
 \end{array} \right)
\eq
with constants $c$ and $s$
will put them into different fibres.
It is not necessary to assume $c^2+s^2=1$, we may find a suitable $c$ and $s$ as an integer or rational number.
For the case at hand $c=1$ and $s=2$ will do the job.
In this way we don't introduce any new variables.
We have verified that after a coordinate transformation
(i) $\det \Omega^{(2)}{}'$ has four critical points, 
(ii) the intersection matrix has rank 4
and (iii) the entries of 
$(A_x)_{i j}$, for $4 \le i,j \le 7$ are computed correctly also in the case where the masses are set equal from the start.

\subsection{Feynman integral reduction}
\label{sect:feynman_integral_reduction}

Intersection numbers are also useful for Feynman integral reductions.
We present here an example, where the use of intersection numbers leads (almost) to a back-of-an-envelope calculation.
\begin{figure}
\begin{center}
\includegraphics[scale=1.0]{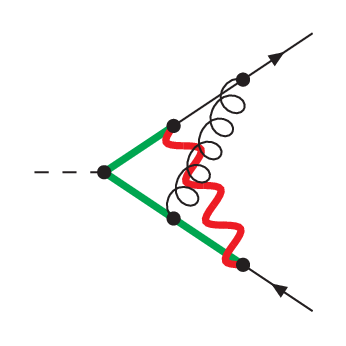}
\end{center}
\caption{
A non-planar Feynman diagram contributing to the mixed ${\mathcal O}(\alpha \alpha_s)$-corrections 
to the decay $H \rightarrow b \bar{b}$ through a $H t \bar{t}$-coupling.
The Higgs boson is denoted by a dashed line, a top quark by a green line, a bottom quark with a black line
and a gluon by a curly line.
Particles with mass $m_W$ are drawn with a wavy line.
}
\label{fig_Higgs_decay}
\end{figure}
Figure~\ref{fig_Higgs_decay} shows a non-planar Feynman diagram contributing to the mixed ${\mathcal O}(\alpha \alpha_s)$-corrections 
to the decay $H \rightarrow b \bar{b}$ through a $H t \bar{t}$-coupling.
The notation follows \cite{Chaubey:2019lum}.
With two independent external momenta and two independent loop momenta we have seven Baikov variables, which we may take as
\begin{align}
 z_1 & = -k_1^2 + m_t^2,
 &
 z_2 & = -\left(k_1-p_1-p_2\right)^2 + m_t^2,
 &
 z_3 & = -\left(k_1+k_2\right)^2,
 \nonumber \\
 z_4 & = -\left(k_1+k_2-p_1\right)^2,
 &
 z_5 & = -k_2^2 + m_W^2,
 &
 z_6 & = -\left(k_2+p_2\right)^2 + m_t^2,
 \nonumber \\
 z_7 & = -\left(k_1-p_1\right)^2 + m_t^2.
 & &
\end{align}
$z_7$ is an auxiliary propagator.
The top sector of the family of Feynman integrals $I_{\nu_1 \nu_2 \nu_3 \nu_4 \nu_5 \nu_6 0}$ has one master integral, which we may takes as
\bq
 I_{1 1 1 1 1 1 0}.
\eq
Suppose we are interested in the decomposition of $I_{1 1 1 1 1 1 (-1)}$ in terms of master integrals:
\bq
 I_{1 1 1 1 1 1 (-1)}
 & = &
 c \;
 I_{1 1 1 1 1 1 0}
 + ...,
\eq
where the dots stand for terms proportional to master integrals in lower sectors.
The coefficient $c$ is computed with the help of intersection numbers as follows:
For the top sector we may work on the maximal cut
$z_1=z_2=z_3=z_4=z_5=z_6=0$.
We set
\bq
 p_1 & = & B\left(0,0,0,0,0,0,z_7\right) \; = \; \frac{1}{16} \left(z_7-p^2\right)^2 \left( z_7 + m_W^2 - m_t^2 \right)^2
\eq
and
\bq
 u  \; = \; p_1^{-\frac{1}{2}-\eps},
 & &
 \omega \; = \; d\ln u.
\eq
As basis of $H^1_\omega$ we take
\bq
 \hat{e}_{1111110}^{(1)} & = & 1.
\eq
The dual basis is then
\bq
 \hat{d}_{1111110}^{(1)} & = & \frac{2 \left(1 + 4 \eps\right)\left(3+4\eps\right)}{\left(1+2\eps\right) \left(p^2+m_W^2-m_t^2\right)^2},
\eq
where $p=p_1+p_2$ denotes the momentum of the Higgs boson.
The integrand of $I_{1 1 1 1 1 1 (-1)}$ on the maximal cut is
\bq
 \hat{e}_{111111(-1)}^{(1)} & = & z_7.
\eq
The sought-after coefficient $c$ is then given by
\bq
 c 
 & = &
 \left\langle e_{111111(-1)}^{(1)} \left| d_{1111110}^{(1)} \right. \right\rangle
 \; = \; \frac{1}{2} \left( p^2 + m_t^2 - m_W^2 \right),
\eq
which agrees with the results from ref.~\cite{Chaubey:2019lum}.

% -----------------------------------------------------------------------------

\section{Conclusions}
\label{sect:conclusions}

In this article I presented an algorithm for the computation of intersection numbers of twisted cocycles,
which avoids in intermediate steps algebraic extensions like square roots.
This is an improvement above the current state-of-the-art.
The algorithm may prove useful in applications towards Feynman integral reductions and the computation of differential
equations for Feynman integrals.

\subsection*{Acknowledgements}

I would like to thank Pierpaolo Mastrolia and Sebastian Mizera for helpful discussions.
I also would like to thank all organisers of the workshop ``MathemAmplitudes 2019: Intersection Theory \& Feynman Integrals'', where this work 
was initiated.

\subsection*{Data Availability}

Data sharing is not applicable to this article as 
no new data were created or analysed in this study.

%------------------------------------------------------------------------------

\begin{appendix}

\section{The Laurent expansions around singular points}
\label{sect:laurent}

In this appendix we review the algorithm of \cite{Mizera:2019gea,Frellesvig:2019uqt}.
The algorithm computes the intersection number
\bq
 \left\langle \varphi_L \right. \left| \varphi_R \right\rangle
 & = &
 \frac{1}{\left(2\pi i\right)^n}
 \int \iota_\omega\left(\varphi_L\right) \wedge \varphi_R,
 \;\;\;\;\;\;\;\;\;
 \left\langle \varphi_L \right| \in H^{({\bf n})}_\omega,
 \;\;\;\;\;\;
 \left| \varphi_R \right\rangle \in \left( H^{({\bf n})}_\omega \right)^\ast
\eq
as follows: 
For $n=0$ we have $\nu_{\bf 0}=1$ and
\bq
 \left\langle e^{({\bf 0})}_1 \right| \; = \; 1,
 \;\;\;\;\;\;\;\;\;
 \left| d^{({\bf 0})}_1 \right\rangle \; = \; 1,
 \;\;\;\;\;\;\;\;\;
 \left\langle e^{({\bf 0})}_1 \right. \left| d^{({\bf 0})}_1 \right\rangle\; = \; 1.
\eq
Hence, the twisted intersection number of the $0$-forms $\varphi_L=\hat{\varphi}_L$ and
$\varphi_R=\hat{\varphi}_R$ is given by
\bq
 \left\langle \varphi_L \right. \left| \varphi_R \right\rangle
 & = &
 \hat{\varphi}_L \hat{\varphi}_R.
\eq
For $n>0$ one expands
the twisted cohomology class $\langle \varphi_L | \in H^{({\bf n})}_\omega$ 
in the basis of $H^{({\bf n-1})}_\omega$:
\bq
 \left\langle \varphi_L \right|
 & = &
 \sum\limits_{j=1}^{\nu_{\bf n-1}}
 \left\langle \varphi^{({\bf n})}_{L,j} \right| \wedge \left\langle e^{({\bf n-1})}_j \right|.
\eq
By recursion we may assume that all intersection numbers involving the variables $z_1,\dots,z_{n-1}$
are already known, therefore it remains to compute the intersection in the variable $z_n$.
One has
\bq
 \left\langle \varphi_L \right. \left| \varphi_R \right\rangle
 & = &
 \sum\limits_{z_0 \in  {\mathcal S}_{\bf n}} 
 \sum\limits_{j=1}^{\nu_{\bf n-1}}
 \mathop{\mathrm{res}}_{z_n=z_0}
 \left(\hat{\psi}^{({\bf n})}_{L,j}
 \left\langle e^{({\bf n-1})}_j \left| \varphi_R \right. \right\rangle \right)
\eq
where $\hat{\psi}^{({\bf n})}_{L,j}$ is determined by
\bq
 \partial_{z_n} \hat{\psi}^{({\bf n})}_{L,j} + \hat{\psi}^{({\bf n})}_{L, i} \Omega^{({\bf n})}_{i j} & = & \hat{\varphi}^{({\bf n})}_{L,j},
\eq
and $\Omega^{({\bf n})}$ is given by eq.~(\ref{def_Omega_left}).
${\mathcal S}_{\bf n}$ is the set of singular points of $\Omega^{({\bf n})}$ in the variable $z_n$, 
including possibly $\infty$.
The function $\hat{\psi}^{({\bf n})}_{L,j}$ need only be computed locally as a Laurent expansion around 
each singular point.
It is at this stage, where algebraic roots enter: The singular points $z_0 \in  {\mathcal S}_{\bf n}$ are given
by the roots of the polynomials appearing in the denominators of the entries of the matrix $\Omega^{({\bf n})}$.

An alternative formulation of the algorithm of \cite{Mizera:2019gea,Frellesvig:2019uqt} exchanges the roles
of $\varphi_L$ and $\varphi_R$ and computes the intersection number by starting from
\bq
 \left\langle \varphi_L \right. \left| \varphi_R \right\rangle
 & = &
 \frac{1}{\left(2\pi i\right)^n}
 \int \varphi_L \wedge \iota_{-\omega}\left(\varphi_R\right),
 \;\;\;\;\;\;\;\;\;
 \left\langle \varphi_L \right| \in H^{({\bf n})}_\omega,
 \;\;\;\;\;\;
 \left| \varphi_R \right\rangle \in \left( H^{({\bf n})}_\omega \right)^\ast.
\eq
One expands $| \varphi_R \rangle \in ( H^{({\bf n})}_\omega )^\ast$ in a basis of $( H^{({\bf n-1})}_\omega )^\ast$:
\bq
 \left| \varphi_R \right\rangle
 & = &
 \sum\limits_{j=1}^{\nu_{\bf n-1}}
 \left| d^{({\bf n-1})}_j \right\rangle
 \wedge
 \left| \varphi^{({\bf n})}_{R,j} \right\rangle.
\eq
We now have
\bq
 \left\langle \varphi_L \right. \left| \varphi_R \right\rangle
 & = &
 -
 \sum\limits_{z_0 \in  {\mathcal S}_{\bf n}} 
 \sum\limits_{j=1}^{\nu_{\bf n-1}}
 \mathop{\mathrm{res}}_{z_n=z_0} \left(
 \left\langle \left. \varphi_L \right|  d^{({\bf n-1})}_j \right\rangle 
 \hat{\psi}^{({\bf n})}_{R,j}
 \right)
\eq
where $\hat{\psi}^{({\bf n})}_{R,j}$ is determined by
\bq
 \partial_{z_n} \hat{\psi}^{({\bf n})}_{R,j} - \Omega^{({\bf n})}_{j k} \hat{\psi}^{({\bf n})}_{R, k} 
 & = & 
 \hat{\varphi}^{({\bf n})}_{R,j},
\eq
and $\Omega^{({\bf n})}$ is given by eq.~(\ref{def_Omega_right}) (or equivalently by eq.~(\ref{def_Omega_left})).

\end{appendix}

%------------------------------------------------------------------------------
% references
{\footnotesize
\bibliography{/home/stefanw/notes/biblio}

\begin{thebibliography}{10}

\bibitem{aomoto1975}
K.~Aomoto,
\newblock J. Math. Soc. Japan {\bf 27}, 248 (1975).

\bibitem{Matsumoto:1994}
K.~Matsumoto,
\newblock Kyushu Journal of Mathematics {\bf 48}, 335 (1994).

\bibitem{cho1995}
K.~Cho and K.~Matsumoto,
\newblock Nagoya Math. J. {\bf 139}, 67 (1995).

\bibitem{matsumoto1998}
K.~Matsumoto,
\newblock Osaka J. Math. {\bf 35}, 873 (1998).

\bibitem{Ohara:2003}
K.~Ohara, Y.~Sugiki, and N.~Takayama,
\newblock Funkcialaj Ekvacioj {\bf 46}, 213 (2003).

\bibitem{Goto:2013}
Y.~{Goto},
\newblock International Journal of Mathematics {\bf 24}, 1350094 (2013),
  arXiv:1308.5535.

\bibitem{Goto:2015aaa}
Y.~Goto and K.~Matsumoto,
\newblock Nagoya Math. J. {\bf 217}, 61 (2015), arXiv:1310.4243.

\bibitem{Goto:2015aab}
Y.~Goto,
\newblock Osaka J. Math. {\bf 52}, 861 (2015), arXiv:1310.6088.

\bibitem{Goto:2015aac}
Y.~{Goto},
\newblock Kyushu Journal of Mathematics {\bf 69}, 203 (2015), arXiv:1406.7464.

\bibitem{Matsubara-Heo:2019}
S.-J. {Matsubara-Heo} and N.~{Takayama},
\newblock (2019), arXiv:1904.01253.

\bibitem{Aomoto:book}
K.~Aomoto and M.~Kita,
\newblock {\em Theory of Hypergeometric Functions} (Springer, 2011).

\bibitem{Yoshida:book}
M.~Yoshida,
\newblock {\em Hypergeometric Functions, My Love} (Vieweg, 1997).

\bibitem{Cachazo:2013gna}
F.~Cachazo, S.~He, and E.~Y. Yuan,
\newblock Phys.Rev. {\bf D90}, 065001 (2014), arXiv:1306.6575.
%%CITATION = ARXIV:1306.6575;%%

\bibitem{Cachazo:2013hca}
F.~Cachazo, S.~He, and E.~Y. Yuan,
\newblock Phys.Rev.Lett. {\bf 113}, 171601 (2014), arXiv:1307.2199.
%%CITATION = ARXIV:1307.2199;%%

\bibitem{Cachazo:2013iea}
F.~Cachazo, S.~He, and E.~Y. Yuan,
\newblock JHEP {\bf 1407}, 033 (2014), arXiv:1309.0885.
%%CITATION = ARXIV:1309.0885;%%

\bibitem{Mizera:2017rqa}
S.~Mizera,
\newblock Phys. Rev. Lett. {\bf 120}, 141602 (2018), arXiv:1711.00469.
%%CITATION = ARXIV:1711.00469;%%

\bibitem{Mizera:2017cqs}
S.~Mizera,
\newblock JHEP {\bf 08}, 097 (2017), arXiv:1706.08527.
%%CITATION = ARXIV:1706.08527;%%

\bibitem{Mizera:2019gea}
S.~Mizera,
\newblock {\em {Aspects of Scattering Amplitudes and Moduli Space
  Localization}},
\newblock PhD thesis, Perimeter Inst. Theor. Phys., 2019, arXiv:1906.02099.
%%CITATION = ARXIV:1906.02099;%%

\bibitem{Mizera:2019blq}
S.~Mizera,
\newblock (2019), arXiv:1912.03397.
%%CITATION = ARXIV:1912.03397;%%

\bibitem{Mastrolia:2018uzb}
P.~Mastrolia and S.~Mizera,
\newblock JHEP {\bf 02}, 139 (2019), arXiv:1810.03818.
%%CITATION = ARXIV:1810.03818;%%

\bibitem{Frellesvig:2019kgj}
H.~Frellesvig {\em et~al.},
\newblock JHEP {\bf 05}, 153 (2019), arXiv:1901.11510.
%%CITATION = ARXIV:1901.11510;%%

\bibitem{Frellesvig:2019uqt}
H.~Frellesvig {\em et~al.},
\newblock Phys. Rev. Lett. {\bf 123}, 201602 (2019), arXiv:1907.02000.
%%CITATION = ARXIV:1907.02000;%%

\bibitem{Mizera:2019vvs}
S.~Mizera and A.~Pokraka,
\newblock JHEP {\bf 02}, 159 (2020), arXiv:1910.11852.

\bibitem{Chen:2020uyk}
J.~Chen, X.~Jiang, X.~Xu, and L.~L. Yang,
\newblock Phys. Lett. B {\bf 814}, 136085 (2021), arXiv:2008.03045.

\bibitem{Frellesvig:2020qot}
H.~Frellesvig {\em et~al.},
\newblock JHEP {\bf 03}, 027 (2021), arXiv:2008.04823.

\bibitem{Caron-Huot:2021xqj}
S.~Caron-Huot and A.~Pokraka,
\newblock (2021), arXiv:2104.06898.

\bibitem{Tkachov:1981wb}
F.~V. Tkachov,
\newblock Phys. Lett. {\bf B100}, 65 (1981).
%%CITATION = PHLTA,B100,65;%%

\bibitem{Chetyrkin:1981qh}
K.~G. Chetyrkin and F.~V. Tkachov,
\newblock Nucl. Phys. {\bf B192}, 159 (1981).
%%CITATION = NUPHA,B192,159;%%

\bibitem{mimachi2004}
K.~Mimachi, K.~Ohara, and M.~Yoshida,
\newblock Tohoku Math. J. (2) {\bf 56}, 531 (2004).

\bibitem{Weinzierl:2014vwa}
S.~Weinzierl,
\newblock JHEP {\bf 1404}, 092 (2014), arXiv:1402.2516.
%%CITATION = ARXIV:1402.2516;%%

\bibitem{Sogaard:2015dba}
M.~Søgaard and Y.~Zhang,
\newblock Phys. Rev. {\bf D93}, 105009 (2016), arXiv:1509.08897.
%%CITATION = ARXIV:1509.08897;%%

\bibitem{Bosma:2016ttj}
J.~Bosma, M.~Søgaard, and Y.~Zhang,
\newblock Phys. Rev. {\bf D94}, 041701 (2016), arXiv:1605.08431.
%%CITATION = ARXIV:1605.08431;%%

\bibitem{Cattani:2005}
E.~Cattani and A.~Dickenstein,
\newblock in: Bronstein M. et al. (eds), Solving Polynomial Equations,
  Algorithms and Computation in Mathematics, vol 14. Springer , 1 (2005).

\bibitem{Zhang:2016kfo}
Y.~Zhang,
\newblock {Lecture Notes on Multi-loop Integral Reduction and Applied Algebraic
  Geometry},
\newblock 2016, arXiv:1612.02249.
%%CITATION = ARXIV:1612.02249;%%

\bibitem{Moser:1959}
J.~Moser,
\newblock Mathematische Zeitschrift {\bf 1}, 379 (1959).

\bibitem{Lee:2014ioa}
R.~N. Lee,
\newblock JHEP {\bf 04}, 108 (2015), arXiv:1411.0911.
%%CITATION = ARXIV:1411.0911;%%

\bibitem{Griffiths:book}
P.~Griffiths and J.~Harris,
\newblock {\em Principles of Algebraic Geometry} (John Wiley \& Sons, New York,
  1994).

\bibitem{Lee:2013hzt}
R.~N. Lee and A.~A. Pomeransky,
\newblock JHEP {\bf 11}, 165 (2013), arXiv:1308.6676.
%%CITATION = ARXIV:1308.6676;%%

\bibitem{Baikov:1996iu}
P.~A. Baikov,
\newblock Nucl. Instrum. Meth. {\bf A389}, 347 (1997), arXiv:hep-ph/9611449.
%%CITATION = HEP-PH/9611449;%%

\bibitem{Lee:2009dh}
R.~N. Lee,
\newblock Nucl. Phys. {\bf B830}, 474 (2010), arXiv:0911.0252.
%%CITATION = 0911.0252;%%

\bibitem{Frellesvig:2017aai}
H.~Frellesvig and C.~G. Papadopoulos,
\newblock JHEP {\bf 04}, 083 (2017), arXiv:1701.07356.
%%CITATION = ARXIV:1701.07356;%%

\bibitem{Bosma:2017ens}
J.~Bosma, M.~Sogaard, and Y.~Zhang,
\newblock JHEP {\bf 08}, 051 (2017), arXiv:1704.04255.
%%CITATION = ARXIV:1704.04255;%%

\bibitem{Harley:2017qut}
M.~Harley, F.~Moriello, and R.~M. Schabinger,
\newblock JHEP {\bf 06}, 049 (2017), arXiv:1705.03478.
%%CITATION = ARXIV:1705.03478;%%

\bibitem{Grozin:2011mt}
A.~G. Grozin,
\newblock Int. J. Mod. Phys. {\bf A26}, 2807 (2011), arXiv:1104.3993.
%%CITATION = ARXIV:1104.3993;%%

\bibitem{Bogner:2019lfa}
C.~Bogner, S.~Müller-Stach, and S.~Weinzierl,
\newblock Nucl. Phys. B {\bf 954}, 114991 (2020), arXiv:1907.01251.

\bibitem{Caffo:1998du}
M.~Caffo, H.~Czyz, S.~Laporta, and E.~Remiddi,
\newblock Nuovo Cim. {\bf A111}, 365 (1998), arXiv:hep-th/9805118.
%%CITATION = HEP-TH/9805118;%%

\bibitem{Chaubey:2019lum}
E.~Chaubey and S.~Weinzierl,
\newblock JHEP {\bf 05}, 185 (2019), arXiv:1904.00382.
%%CITATION = ARXIV:1904.00382;%%

\end{thebibliography}
\bibliographystyle{/home/stefanw/latex-style/h-physrev5}
}

\end{document}